\documentclass[journal,twocolumn]{IEEEtran}
\usepackage{diagbox} % 加载宏包
\usepackage{graphicx}
\usepackage{epstopdf}
\usepackage{psfrag}
\usepackage{subfigure}
\usepackage{url}
\usepackage{stfloats}
\usepackage{amsfonts,amssymb,amsmath,bm,paralist,theorem,cite,ifthen,color,nccmath}

\usepackage{array}
\usepackage{caption}
\usepackage{calc}
\usepackage{multirow}
\usepackage[ruled]{algorithm2e}
\usepackage{setspace}
\usepackage{float}
\usepackage{soul}
\usepackage{tikz}
\usepackage{wrapfig}
\usepackage[english]{babel}
\usepackage{rotfloat}
\usepackage{mathtools}
\usepackage{makecell}
\usepackage{booktabs}
\usepackage{colortbl}
\usepackage[normalem]{ulem}
\usepackage{bbm}
\usepackage{cases}
\usepackage{threeparttable}
\usepackage{titlesec}
\usepackage[]{xcolor}
\newcommand{\blue}[1]{\textcolor{blue}{#1}}

\newtheorem{proposition}{Proposition}

\newcommand{\qedsymbol}{\hfill \(\blacksquare\)}

\graphicspath{{Figures/}}
\newcommand\Ccl{\ensuremath{\mathcal{C}}}
\newcommand\Ncl{\ensuremath{\mathcal{N}}}
\newcommand\Ocl{\ensuremath{\mathcal{O}}}

\newcommand\Bcl{\ensuremath{\mathcal{B}}}

\newcommand\Scl{\ensuremath{\mathcal{S}}}
\newcommand\Icl{\ensuremath{\mathcal{I}}}
\newcommand\Acl{\ensuremath{\mathcal{A}}}

\newcommand\Fcl{\ensuremath{\mathcal{F}}}
\newcommand\Lcl{\ensuremath{\mathcal{L}}}

\newcommand\Cs{\ensuremath{{\mathbb{C}}}}
\newcommand\Es{\ensuremath{{\mathbb{E}}}}
\newcommand\Rs{\ensuremath{{\mathbb{R}}}}

\newcommand\Ab{\ensuremath{ \mathbf{A} }}

\newcommand\Cb{\ensuremath{ \mathbf{C} }}
\newcommand\Db{\ensuremath{ \mathbf{D} }}
\newcommand\Fb{\ensuremath{ \mathbf{F} }}

\newcommand\Hb{\ensuremath{ \mathbf{H} }}
\newcommand\Ib{\ensuremath{ \mathbf{I} }}
\newcommand\Tb{\ensuremath{ \mathbf{T} }}

\newcommand\Ub{\ensuremath{ \mathbf{U} }}

\newcommand\Vb{\ensuremath{ \mathbf{V} }}
\newcommand\Wb{\ensuremath{ \mathbf{W} }}
\newcommand\Jb{\ensuremath{ \mathbf{J} }}

\newcommand\Sb{\ensuremath{ \mathbf{S} }}

\newcommand\Yb{\ensuremath{ \mathbf{Y} }}
\newcommand\Zb{\ensuremath{ \mathbf{Z} }}

\newcommand\ab{\ensuremath{ \mathbf{a} }}

\newcommand\fb{\ensuremath{ \mathbf{f} }}
\newcommand\hb{\ensuremath{ \mathbf{h} }}
\newcommand\nb{\ensuremath{ \mathbf{n} }}

\newcommand\qb{\ensuremath{ \mathbf{q} }}
\newcommand\rb{\ensuremath{ \mathbf{r} }}
\newcommand\ssb{\ensuremath{ \mathbf{s} }}

\newcommand\wb{\ensuremath{ \mathbf{w} }}
\newcommand\xb{\ensuremath{ \mathbf{x} }}
\newcommand\yb{\ensuremath{ \mathbf{y} }}

\newcommand\Gammab{\ensuremath{{\bm \Gamma}}}

\newcommand\lambdab{\ensuremath{{\bm \lambda}}}
\newcommand\Omegab{\ensuremath{{\bm \Omega}}}

\newcommand\Xib{\ensuremath{{\bm \Xi}}}

\newcommand\fc{\ensuremath{ f_{\rm c} }}
\newcommand\ghz{\ensuremath{ {\rm GHz} }}
\newcommand\Nt{\ensuremath{ N_{\rm T} }}
\newcommand\Nu{\ensuremath{ N_{\rm U} }}
\newcommand\Nr{\ensuremath{ N_{\rm R} }}
\newcommand\Ns{\ensuremath{ N_{\rm S} }}
\newcommand\Nrf{\ensuremath{ N_{\rm RF} }}

\newcommand\Ndt{\ensuremath{ N_{\rm td}^{\rm t} }}
\newcommand\Ndr{\ensuremath{ N_{\rm td}^{\rm r} }}
\newcommand\Nfdt{\ensuremath{ N_{\rm ftd}^{\rm t} }}
\newcommand\Nfdr{\ensuremath{ N_{\rm ftd}^{\rm r} }}
\newcommand\ps{\ensuremath{ {\rm ps} }}

\newcommand\BSR{\ensuremath{ {\rm BSR} }}
\newcommand\Frf{\ensuremath{ \mathbf{F}_{\rm RF} }}

\newcommand\Fbb{\ensuremath{ \mathbf{F}_{\rm BB} }}

\newcommand\mW{\textrm{ mW }}
\newcommand\dB{\textrm{ dB }}
\newcommand\SNR{\textrm{ SNR}}
\newcommand\SINR{\textrm{ SINR }}
\newcommand\Nnb{\ensuremath{ N_{\rm nb} }}

\newcommand\diag{\ensuremath{{\rm diag}}}
\newcommand\tr{\ensuremath{{\rm Tr}}}

\newcommand\rank{\ensuremath{{\rm rank}}}

\titlespacing{\section}{2pt}{15pt}{7pt}
\titlespacing{\subsection}{1pt}{4pt}{4pt}
\IEEEoverridecommandlockouts

\begin{document}

\title{Switch-based Hybrid Beamforming Transceiver Design for Wideband Communications with Beam Squint}

\author{Mengyuan~Ma,~\IEEEmembership{Student Member,~IEEE}, Nhan~Thanh~Nguyen,~\IEEEmembership{Member,~IEEE}, and~Markku~Juntti,~\IEEEmembership{Fellow,~IEEE}

\thanks{This paper was presented in part in Proc. ITG Workshop Smart Antennas 2021 \cite{ma2021switch}.}

\thanks{This work has been supported in part by Academy of Finland under 6G Flagship (Grant No. 318927) and EERA Project (Grant No. 332362).}
\thanks{Mengyuan Ma,  Nhan Thanh Nguyen, and Markku Juntti are with the Centre for Wireless Communications, University of Oulu, 90014 Oulu, Finland (e-mail:\{mengyuan.ma, nhan.nguyen, markku.juntti\}@oulu.fi).}
}

\maketitle

\vspace{-10mm}
\begin{abstract}
Hybrid beamforming (HBF) transceiver architectures based on frequency-independent phase shifters (PSs) are sensitive to phases and physical directions, resulting in limited capability to compensate for the detrimental effects of the beam squint. Motivated by the fact that switches are phase-independent and more power/cost efficient than PSs, we consider the switch-based HBF (SW-HBF) for wideband large-scale multiple-input multiple-output systems in this paper. We first derive a closed-form expression of the beam squint ratio unveiling that the severity of beam squint linearly increases with the number of antennas, the antenna spacing distance, and the fractional bandwidth. We then focus on the SW-HBF designs to maximize the spectral efficiency (SE) in both single-user (SU) and multiuser (MU) systems. The formulated problems in both cases exhibit intractable non-convex and mixed-integer challenges. To address them, for the SU case, by combining the tabu search (TS) method and projected gradient ascend (PGA), we propose an efficient heuristic PGA-TS algorithm to design analog beamformers while the digital ones admit closed-form solutions. For the MU case, we develop a two-step algorithm based on fractional programming and the PGA-TS method. Simulation results show that the proposed SW-HBF schemes are efficient and can outperform PS-based HBF architectures in terms of both SE and energy efficiency in wideband systems.
\end{abstract}


\IEEEpeerreviewmaketitle

\section{Introduction} \label{sec:introduction}
\IEEEPARstart{W}{ideband} communication systems are promising to meet the ever-increasing demand for ultra-high-speed data rates of future wireless networks \cite{jiang2021road}. Millimeter wave (mmWave) communications thereby have been considered essential for future wireless communication systems \cite{ahmed2018survey}. %However, due to the high frequencies, mmWave communications can suffer from severe path loss \cite{ahmed2018survey}. To overcome this, very directive and high-gain antennas are required. %When beam direction also needs to be flexibly adjusted, phased antenna arrays can be employed to generate and steer the high-gain directional beams \cite{chen2019survey}.
The short wavelength of mmWave signals allows compact deployment of large numbers of antenna elements, facilitating the implementation of large-scale or massive multiple-input multiple-output (MIMO) systems to compensate for severe path loss \cite{mumtaz2016mmwave}. However, a large-scale MIMO system deploying the conventional fully digital beamforming (DBF) architectures requires excessively large numbers of power-hungry radio frequency (RF) chains causing prohibitive power consumption and hardware costs \cite{buzzi2018energy}. Such limitations pose significant challenges to the system. Therefore, hybrid beamforming (HBF) architectures have been proposed to divide the overall beamformer into a high-dimensional analog beamformer and a lower-dimensional digital beamformer \cite{zhang2005variable}. Such architectures allow a low number of digital and RF branches while guaranteeing multiplexing gain, significantly reducing power consumption and implementation complexity \cite{nguyen2022hybrid}. 

\subsection{Prior Works}\label{sec:introduction prior works}
Driven by the tradeoff between spectral efficiency (SE) and energy efficiency (EE), various HBF architectures have emerged based on the fact that the analog beamformer can be realized by a network of phase shifters (PSs) \cite{el2014spatially,sohrabi2016hybrid,gao2016energy,lin2019transceiver,zhao2021partially}, switches (SWs) \cite{mendez2016hybrid,nosrati2021online,nosrati2022switch}, or both \cite{payami2018hybrid,kaushik2019dynamic,nguyen2019unequally,gadiel2021dynamic}. Specifically,  it was shown in \cite{mendez2016hybrid} that PS-based HBF (PS-HBF) architectures can reap more SE than SW-based HBF (SW-HBF) architectures, while the latter has more advantages of EE assuming the same connection style used in the analog beamformer. Inspired by this, the combination of PSs and SWs for HBF was investigated in \cite{payami2018hybrid,kaushik2019dynamic,nguyen2019unequally,gadiel2021dynamic}, which shows such HBF design can attain a better tradeoff between the SE and EE than the purely SW-HBF or PS-HBF schemes. The aforementioned works mainly focus on narrowband systems. In contrast, HBF design for wideband systems is more challenging because the frequency-flat analog beamformer must be shared across the whole signal bandwidth or all the subcarriers in the multicarrier case.  To tackle the challenges, efficient PS-HBF algorithms have been proposed in \cite{park2017dynamic,sohrabi2017hybrid,chen2019channel,ma2021closed} by searching the eigenvector spaces and constructing the common analog beamformer across all frequencies. However, the detrimental \textit{beam squint} effect \cite{cai2016effect}, which can significantly degrade the performance of wideband systems, was ignored in those works. In conventional narrowband systems, the analog beamformer aligned with the physical direction at the central frequency can almost achieve the maximum array gain for signals at other frequencies. However, in wideband systems, due to frequency-independent PSs, the beam aligned with the central frequency squints to the incident angles of the incoming signals at other frequencies, resulting in non-negligible array gain loss, termed the beam squint effect \cite{cai2016effect}. Notably, this effect becomes more severe for systems with wider bandwidth and more antennas \cite{dai2022delay}. As such, the performance of the HBF strategies without considering the beam squint issue can significantly deteriorate in mmWave and THz communications, as shown in \cite{gao2021wideband,dai2022delay}.

To mitigate the beam squint effect, several approaches have been proposed in the literature \cite{cai2016effect,liu2018space,dai2022delay,gao2021wideband, wu20223,beam2023maicc,ozen2023interference}. Specifically, wide beams are designed in \cite{cai2016effect,liu2018space} to adapt for large bandwidth, which can reduce the loss in array gain induced by the beam squint effect. However, since the beam squint effect is essentially caused by the time delay spread over the antenna array \cite{wang2018spatial}, it cannot be eliminated solely relying on the frequency-independent phased array. Embedding true-time-delayers (TTDs) into the RF front-end can introduce frequency-dependent phase shifts, alleviating the beam squint. Therefore, TTD-aided HBF architectures have been proposed in \cite{dai2022delay,gao2021wideband} wherein TTDs are inserted between the RF chain and the network of PSs to provide additional delay rather than replacing the PSs. While the approach has been shown to effectively eliminate the beam squint, the implementation can be a significant challenge in practice. Since the time delay over the antenna array increases with signal bandwidth and array size, more on-chip TTDs are required for compensation \cite{ozen2023interference}, which brings forth the scalability challenges regarding insertion losses, power consumption, chip area, and linearity \cite{cao2015advanced}. For example, in a system operated at $300~\ghz$ with a 256-antenna uniform linear array (ULA), the delay spread over the antennas could be up to $426~\ps$ \cite{dai2022delay} and the required TTD could have a power consumption of $285~\mW$ \cite{cho2018true}. Furthermore, TTDs providing large variable delays can cause significant insertion loss at high frequencies, which can be greater than $40~\dB$ for a delay of $400~\ps$ at $20~\ghz$ \cite{hu20151}. The above facts have motivated the deployment of uniform planar array (UPA) \cite{nguyen2022beam,wu20223,beam2023maicc} to diminish the delays across the antenna elements, alleviating the demands for TTDs. Additionally, the intelligent transmission surface used in the HBF transmitter for beamforming \cite{wu2023hybrid} functions like a UPA and can also alleviate the beam squint. Furthermore, fixed-time TTDs (FTTDs) \cite{yan2022energy} can be used to further reduce power consumption. However, FTTDs still share the large variable delay and insertion losses as conventional TTDs do. 

%Compared to the TTD and PS-based analog architecture discussed above, a SW-based one has lower power consumption and insertion loss \cite{mendez2016hybrid}. For example, the power consumption and insertion loss of a switch can be as low as $5~\mw $ and $1~\dB$ \cite{mendez2016hybrid}, while those can be up to $42~\mw $ and $15~\dB$ for a PS \cite{kim2014220} and $285~\mw $ and $40~\dB$ for a TTD \cite{cho2018true,hu20151}, respectively. Moreover, employing only SWs simplifies the RF chain design and allows its fast adjustment to rapid channel variations \cite{nosrati2021online}. Furthermore, additional power amplification for offsetting significant insertion loss or increase of the circuit area may be avoided. \blue{In particular, SW-HBF architectures were shown in our previous work \cite{ma2021switch} to be less affected by the beam squint effect and can achieve better SE and EE than PS-HBF architectures in wideband systems, which is not seen in conventional narrowband systems. Those advantages of SW-HBF could make it another promising technique in wideband systems, which, to the best of the author's knowledge, has not been studied in existing literature. We are motivated to fill this blank by thoroughly investigating the performance of SW-HBF architectures in THz communications.}

% \vspace{-5mm}
\subsection{Contributions}\label{sec:contributions}

% Compared to the TTD and PS-based analog architectures discussed above, a SW-based one has lower power consumption and insertion loss \cite{mendez2016hybrid}. Moreover, employing only SWs simplifies the RF chain design and allows its fast adjustment to rapid channel variations \cite{nosrati2021online}. Furthermore, additional power amplification for offsetting significant insertion loss or increase of the circuit area may be avoided. In particular, consisting of switches, SW-HBF architectures are less affected by the beam squint effect. Therefore, they can achieve better SE and EE performances than PS-based HBF architectures, especially under severe beam squint effect in wideband systems, which are distinct from the cases of narrowband systems. However, the advantage of SW-HBF architecture in wideband systems has not been investigated in the literature.

% In this paper, we consider wideband single-user and multiuser massive MIMO systems. Through the analysis of the beam squint effect and SW-HBF designs to mitigate it, we show that SW-HBF architectures are more robust to beam squint compared to the PS-HBF counterpart. Considering their low cost and low power consumption, SW-HBF is demanding for wideband massive MIMO systems. 

In contrast to the TTD and PS-based analog architectures discussed above, a SW-based architecture boasts lower power consumption and insertion loss \cite{mendez2016hybrid}. The exclusive use of switches simplifies the design of the RF chain and facilitates rapid adjustment to channel variations \cite{nosrati2021online}. Specifically, the hardware implementation of the adjustable TTD typically involves numerous switches and other components \cite{yan2022energy}, especially for large-scale systems. In contrast to the TTD-based architectures, SW-based ones obviate the necessity for additional power amplification to counteract insertion losses, thus, significantly reducing circuital costs. Furthermore, switches designed for mmWave systems can operate at speeds on the order of nanoseconds or even sub-nanoseconds \cite{li2020dynamic}, which could be far less than the channel coherence time. For example, the coherence time for a carrier frequency of 110 GHz is approximately $0.2$ ms \cite{han2022terahertz}. Therefore, SW-based analog architectures are promising to meet the requirements of practical implementations. 
Moreover, SW-HBF architectures are less affected by the beam squint effect than PS-HBF ones. This characteristic contributes to improved SE and EE performance in wideband systems. However, to date, the advantages of SW-HBF architectures in wideband systems remain unexplored in the literature.

In this paper, our focus extends to wideband single-user and multiuser massive MIMO systems. Through an in-depth analysis of the beam squint effect and carefully designed SW-HBF, we show that SW-HBF architectures exhibit enhanced resilience to beam squint when compared to their PS-HBF counterparts. Given their cost effectiveness and low power consumption, SW-HBF architectures emerge as a compelling choice for wideband massive MIMO systems. Our specific \textbf{contributions} can be summarized as follows:
\begin{itemize}

    \item We first obtain a closed-form expression of the beam squint ratio (BSR) that quantifies the severity of the beam squint effect in wideband MIMO systems \cite{dai2022delay}. %Compared to the analysis in \cite{dai2022delay}, we further take the antenna spacing into consideration. 
    The analytical results show that the severity of the beam squint effect increases linearly with the number of antennas, the antenna spacing distance, and the fractional bandwidth. This reveals that reducing the spacing of the antenna could be a promising approach to mitigate the beam squint. Moreover, our analysis shows that a SW-based beamformer may achieve a higher expected array gain than a PS-based one.
    
    \item We then focus on joint transmit precoding and receive combining in SW-HBF architectures for single-user (SU) MIMO communications, aiming to maximize the average SE. The formulated problems are challenging due to the non-convex mixed-integer nature. To overcome the challenge, we propose a projected gradient ascent-aided (PGA) tabu search (TS) algorithm, which finds an efficient solution to the analog beamformer, while the digital beamformers are obtained with closed-form solutions based on the solved analog ones.

    \item Furthermore, we study the hybrid precoding design in multiuser (MU) multiple-input single-output (MISO) downlink, aiming to maximize the average sum rate. To address the challenging non-convex mixed-integer problem, we propose an efficient two-step algorithm wherein the digital baseband beamformers and the relaxed analog one are first obtained with the fractional programming framework. Then an efficient analog beamformer is found by performing the PGA-TS algorithm based on the relaxed one.

    %\item \blue{Furthermore, we study the performance of SW-HBF in a multiuser multiple-input single-output (MU-MISO) system. The problem is formulated to maximize the average weighted sum rate by jointly optimizing the analog and digital precoder. To address the mixed-integer and nonconvex challenge, we propose an efficient two-step algorithm, wherein the digital precoder and relaxed analog one are first solved with the fraction programming theory. Then, the proposed PGA-TS algorithm for SU-MIMO systems is similarly utilized to find an efficient analog precoder based on the relaxed one.}
    
    \item Extensive numerical simulations are provided to evaluate the performance of the proposed SW-HBF algorithms in both the considered SU and MU systems. The numerical results demonstrate that the proposed SW-HBF schemes perform close to the optimal exhaustive search (ES) scheme with significantly reduced complexity. Moreover, it is shown that under a severe beam squint effect, the SW-HBF scheme can outperform the TTD-based and PS-based HBF schemes.  
    
\end{itemize}

To focus on the analysis and design of SW-HBF schemes, we assume the availability of full channel state information (CSI) at both the transmitter and the receiver \cite{yu2016alternating,sohrabi2017hybrid,li2020dynamic,yan2022energy}. Thereby we gain a deep understanding of the fundamental performance limits of beam squint in wideband HBF architectures. Moreover, it is shown that CSI can be estimated with the compressed sensing method \cite{mendez2016hybrid}. We note that this work is an extension of our previous conference paper \cite{ma2021switch}, where we only considered the receiver design for SU-MIMO systems. We herein not only investigate the joint designs of both transmit precoding and receive combining for SU-MIMO but also study the MU-MISO downlink precoding design. Moreover, we also present a theoretical analysis of the beam squint effect, a detailed analysis of the computational complexity of the proposed SW-HBF algorithms, and extensive numerical simulation results. 
% \vspace{-5mm}
\subsection{Organization and notation}

The remainder of this paper is structured as follows. We commence with the presentation of the signal and channel models in Section \ref{sc:system model}. In Section \ref{sc:beam squint effect analysis}, we delve into an analysis of the beam squint effect. Subsequently, we introduce SW-HBF designs for SU-MIMO and MU-MISO systems in Section \ref{sc:SW-HBF in SU-MIMO} and Section \ref{sec:SW-HBF for MU-MISO}, respectively. Section \ref{sc:simulation results} provides our simulation results, and we conclude this paper in Section \ref{sc:conclusions}.

Throughout this paper, scalars, vectors, and matrices are denoted by the lower-case, boldface low-case, and boldface upper-case letters, respectively. $\left(\cdot\right)^T$, $\left(\cdot\right)^H$, and $\left(\cdot\right)^{\dagger}$ represent the transpose, conjugate transpose, and Moore–Penrose inverse operations. $\left\|\cdot\right\|_p$ and $\left\|\cdot\right\|_F$ denote the $p$-norm and Frobenius norm. The expectation operator is represented by $\Es\left(\cdot\right)$. Furthermore, we use $\Acl$ to denote a set and $\left|\Acl \right|$ to denote the cardinality of $\Acl$. In addition, $\left|a \right|$ and $\left|\Ab \right|$ denote the absolute value of the scalar $a$ and the determinant of the matrix $\Ab$. Moreover, complex and binary spaces of size $m \times n$ are represented by $\Cs^{m\times n}$ and $\Bcl^{m\times n}$, respectively. Finally, $\Ib_N$ represents an identity matrix of size $N\times N$, and $\Ab(i,j)$ denotes the element of $\Ab$ in the $i$-th row and the $j$-th column. 

%$\Cc\Nc(\boldsymbol{\mu}, \Cb)$ denotes the distribution of a circularly symmetric complex Gaussian (CSCG) random vector variable with mean $\boldsymbol{\mu}$ an covariance matrix $\Cb$; 
% \vspace{-5mm}
\section{System Model}\label{sc:system model}
%   \begin{figure}[htbp]
%        \centering
%        \subfigure[Phase shifter based hybrid combining.]{\label{fig: PS HBF} \includegraphics[width=0.7\textwidth]{system.JPG}}
%        \subfigure[Switch based hybrid combining]{\label{fig: SW HBF} \includegraphics[width=0.7\textwidth]{SW_HBF.png}}
%        	\caption{Illustration of hybrid combining structure.}
%        	\label{fig:Illustration of hybrid combining structure}
%        	\vspace{-0.3cm}
%  \end{figure}
% \vspace{-5mm}
\subsection{Signal Model}
   \begin{figure*}[t]
   % \vspace{-0.5cm}
        \centering	
         \includegraphics[width=0.9\textwidth]{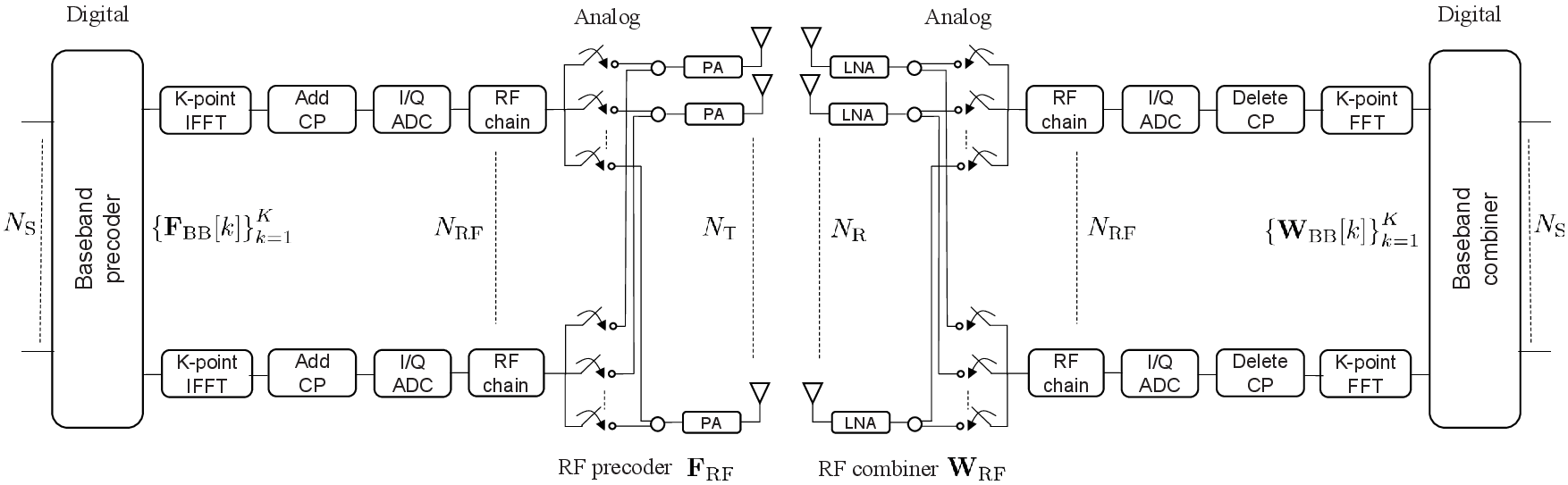}
         \captionsetup{font={small}}
         \vspace{-0.2cm}
        \caption{SW-HBF transceiver structure in a wideband MIMO-OFDM system}
        	\label{fig:HBF transceiver structure in mmWave MIMO-OFDM system}
        	\vspace{-0.5cm}
  \end{figure*}
We consider a wideband SU-MIMO orthogonal frequency-division multiplexing (OFDM) system illustrated in Fig.\ \ref{fig:HBF transceiver structure in mmWave MIMO-OFDM system} wherein the analog beamformers are realized by the fully connected SW-based networks. The transmitter (Tx) equipped with $N_{\rm T}$ antennas simultaneously transmits $N_{\rm{S}}$ data streams to the receiver (Rx) equipped with $N_{\rm R}$ antennas. For ease of exposition and without loss of generality, we assume that both the Tx and Rx are equipped with $N_{\rm RF}$ RF chains and $N_{\rm{S}} \leq N_{\rm RF} \ll \min(N_{\rm T},N_{\rm R})$. Denote by $K$ the number of subcarriers and by $\ssb_k \in \Cs^{N_{\rm S} \times 1}$ the signal vector transmitted via the $k$-th subcarrier, $\Es\left[\ssb_k\ssb^H_k\right]=\Ib_{N_{\rm S}},\; k=1,\cdots,K$. At the Tx, $\ssb_k$ is first precoded by a digital precoder $\Fb_{\rm BB}[k] \in \Cs ^{N_{\rm RF} \times N_{\rm S}}$ in the frequency domain, which is then transformed into the time domain by the inverse fast Fourier transforms (IFFTs), and up-converted to the RF domain using the RF chains. The RF signals are then further processed by an analog precoder $\Fb_{\rm RF}\in \Bcl ^{N_{\rm T} \times N_{\rm RF}}$. At the Rx, the received signal vector is first combined by an analog combiner $\Wb_{\rm RF}\in \Bcl ^{N_{\rm R} \times N_{\rm RF}}$. Note that both the analog precoder $\Fb_{\rm RF}$ and combiner $\Wb_{\rm RF}$ are frequency-flat, i.e., they are used to process all the $K$ subcarriers. After performing fast Fourier transforms (FFTs), the combined signal is further combined in the frequency domain by a baseband combiner $\Wb_{\rm BB}[k] \in \Cs ^{N_{\rm RF} \times N_{\rm S}}$ for each subcarrier. Let $\Hb_k\in \Cs^{N_{\rm R}\times N_{\rm T}}$ be the channel at the $k$-th subcarrier, and $\Fb_k = \Fb_{\rm RF}\Fb_{\rm BB}[k]$, $\| \Fb_k \|^2_F \leq P_{\rm b}, \forall k$, with $P_{\rm b}$ being the power budget for each subcarrier. The received signal at the $k$-th subcarrier is given as
\begin{equation}\label{eq:received signal}
% \vspace{-1mm}
% \small
 % \setlength{\abovedisplayskip}{4pt}
 % \setlength{\belowdisplayskip}{4pt}
    \yb_k= \Wb_k^H  \Hb_k \Fb_k \ssb_k + \Wb_k^H \nb_k,
\end{equation}
where $\Wb_k = \Wb_{\rm RF}\Wb_{\rm BB}[k]$, and $\nb_k\sim \Ccl\Ncl(\boldsymbol{0}, \sigma_{\rm n}^2\Ib_{N_{\rm R}})$ is the additive white Gaussian noise (AWGN) vector with $\sigma_{\rm n}^2$ being its variance.

% \vspace{-5mm}
\subsection{Wideband Channel Model}\label{sc:channel model}
%In the THz-band communications, signal propagation experiences high attenuation and very limited scattering due to the short wavelength and molecular absorption \cite{tarboush2021teramimo}. Therefore, 
We adopt a statistical tap-delay profile modeling the impulse response and multipath parameters for ultra-broadband channels \cite{han2018propagation}. Assuming that ULAs are employed\footnote{In this paper, we consider the ULA for simplicity, but the analysis of the beam squint effect can be extended to the UPA \cite{beam2023maicc}. Besides, the proposed SW-HBF schemes are applicable to antenna arrays of arbitrary geometry.}, the $d$-th tap of the channel at frequency $f$ can be modeled as \cite{alkhateeb2016frequency,li2020dynamic}
\begin{equation}\label{eq:d delay channel}
% \vspace{-2mm}
   \Hb_f[d]= \zeta\sum\limits_{l=1}^{L_{\rm p}} \alpha_l p\left(d T_{\rm s}-\tau_l\right) \ab_{\rm r}\left( \theta_l^{\rm r}, f\right) \ab_{\rm t}^H\left( \theta_l^{\rm t}, f\right), 
\end{equation}
% \vspace{-2mm}
where $\zeta \triangleq \sqrt{\frac{N_{\rm T}N_{\rm R}}{L_{\rm p}}}$ with  $L_{\rm p}$ being the number of distinct scattering paths, and $\alpha_l \sim \Ccl\Ncl(0,1)$ and $\tau_l$ are the complex channel gain and the delay of the $l$-th path, respectively; $\theta_l^r$ and $\theta_l^t$ represent the angles of arrival/departure (AoA/AoD) of the $l$-th path, respectively; $T_{\rm s}$ is the sampling period; $p(\tau)$ denotes the pulse-shaping filter. In \eqref{eq:d delay channel}, $\ab_{\rm r}\left( \theta_l^{\rm r}, f\right) $ and $\ab_{\rm t}\left( \theta_l^{\rm t}, f\right)$ denote the transmit and receive array response vectors at frequency $f$ for the $l$-th path with AoA $\theta_l^{\rm r}$ and AoD $\theta_l^{\rm t}$, respectively. The array response vector $\ab(\theta,f) \in \{\ab_{\rm r}\left( \theta_l^{\rm r}, f\right), \ab_{\rm t}\left( \theta_l^{\rm t}, f\right)\}$ is given as
% \vspace{-2mm}
\begin{equation}\label{eq:array steer vector for MIMO}
\small
   \ab(\theta,f)=\frac{1}{\sqrt{N}}\left[1,e^{-j 2 \pi \Delta \frac{f}{f_{\rm c}}\sin(\theta) },\ldots, e^{-j 2 \pi (N-1)  \Delta \frac{f}{f_{\rm c}}\sin(\theta) }  \right]^T, 
\end{equation}
%\vspace{-2mm}
with $N \in \{N_{\rm T}, N_{\rm R} \}$,  $\theta \in \{ \theta_l^{\rm r}, \theta_l^{\rm t} \}$. Here, $\Delta \triangleq \frac{ d_{\rm a} }{ \lambda_{\rm c} }$ denotes the normalized antenna spacing where $d_{\rm a}$ and $\lambda_{\rm c}$ represent the physical antenna spacing and the wavelength of the central frequency $f_{\rm c}$, respectively.  %  $\lambda_{\rm c}$ denotes the wavelength of carrier frequency $f_{\rm c}$.
  The frequency-domain channel at the $k$-th subcarrier can be expressed as $\Hb_k=\sum_{d=0}^{D-1} \Hb_{f_k}[d] e^{-j\frac{2\pi k}{K}d}$ \cite{tse2005fundamentals}\blue{,}
where $D$ specifies the maximum delay spread. Here, $f_{k}\triangleq f_{\rm c}+\left(k-\frac{K+1}{2} \right)\frac{B}{K},\ \forall k$ denotes the frequency of the $k$-th subcarrier with $B$ denoting the total system bandwidth.

\vspace{-2mm}
\section{Analysis of the Beam Squint Effect}\label{sc:beam squint effect analysis} 

As shown in \cite{dai2022delay}, the beam squint effect can cause significant performance degradation in wideband systems. The beam squint ratio is defined in \cite{dai2022delay} to quantify its severity so that a larger BSR implies a more severe effect. While it is a useful metric, the definition of the BSR involves integration, which is complex to compute and not straightforward to get insights into the factors causing the phenomenon. %In this section, we derive a closed-form expression of the BSR revealing key factors determining the severity of the beam squint. 
In this section, we first derive a closed-form expression of the BSR. We then analyze the expected array gains (EAGs) for both PS-based and SW-based beamformers to evaluate their corresponding relationship with the BSR.
% \vspace{-3mm}
\subsection{Closed-form BSR}\label{sec:beam squint ratio}

The BSR is defined as the expectation of the ratio between the physical direction deviations and half of the beamwidth for all subcarrier frequencies and physical directions \cite{dai2022delay}. Specifically, let $N$ denote the number of antennas in a ULA, and let $\xi_k\triangleq \frac{f_k}{f_{\rm c}}$ and $\vartheta \triangleq \sin \theta$, where $\theta$ denotes the AoA/AoD. Because the half-beamwidth is $\frac{1}{N\Delta}$ \cite{tse2005fundamentals}, and the physical direction deviation for the $k$-th subcarrier can be represented by $\left|\left(1-\xi_k\right)\vartheta \right|$, the BSR is expressed as \cite{dai2022delay}  
\vspace{-1mm}
\begin{equation}\label{eq:BSR definition}
    \text{ BSR}\triangleq \frac{1}{2} \int^{1}_{-1} \frac{1}{K} \sum\limits_{k=1}^K \frac{\left|\left(1-\xi_k\right)\vartheta \right|}{\frac{1}{N\Delta}} d\vartheta.
\end{equation}
Note that  the BSR was originally defined in \cite{dai2022delay} for half-wavelength antenna spacing. We herein take the normalized antenna spacing $\Delta$ into consideration in \eqref{eq:BSR definition}. We present below a more insightful approximation of the BSR. 
% \vspace{-2mm}
\begin{proposition}\label{remark:BSR equation}
 The BSR can be approximately given as
 \begin{equation}\label{eq:BSR closed-form}
   \text{ BSR} \approx \frac{N b\Delta}{8},
   \end{equation}
where $N$, $b$, and $\Delta$ denote the number of antennas, fractional bandwidth, i.e., $b\triangleq \frac{B}{f_{\rm c}}$, and the normalized antenna spacing, respectively. The approximation in \eqref{eq:BSR closed-form} is tight when the number of subcarriers is sufficiently large.
\end{proposition}

{\it Proof:} Recalling that $f_{k}\triangleq f_{\rm c}+\left(k-\frac{K+1}{2} \right)\frac{B}{K},\ \forall k$, we obtain 
% By dividing both sides of \eqref{eq:subcarrier difinition} to $f_{\rm{c}}$, we obtain
\begin{align}\label{eq:xi bandwidth}
    \xi_k = \frac{f_k}{f_{\rm{c}}}&=1+ \left(\frac{k}{K}-\frac{K+1}{2K} \right) b,
\end{align}
where $b\triangleq \frac{B}{f_{\rm c}}$. Therefore, we can write
\begingroup
\allowdisplaybreaks
\begin{align}
        \text{ BSR}  &=\frac{N\Delta b }{2K}  \sum\limits_{k=1}^K \left|\frac{k}{K}-\frac{K+1}{2K}\right|\int^{1}_{-1} | \vartheta|d\vartheta \nonumber \\
        &=\frac{N\Delta b }{2} \cdot \frac{1}{K}\sum\limits_{k=1}^K \left|\frac{k}{K}-\frac{K+1}{2K}\right| \notag \\
        &\overset{(a)}{\approx} \frac{N\Delta b }{2} \int^{1}_{0} \left| x-\frac{1}{2} \right|dx 
        =\frac{N\Delta b}{8}, \label{BSE_approx}
\end{align}
\endgroup
where approximation $(a)$ follows from the facts that $K \gg 1$ and $\int^{1}_{0} f(x)dx=\displaystyle \lim_{K \to \infty} \frac{1}{K} \sum\limits_{k=1}^K f\left( \frac{k}{K} \right)$ for a continuous real-valued function $f(\cdot)$ defined on the closed interval $[0, 1]$ \cite{thomas2005thomas}. \qedsymbol %\epr

   \begin{figure}[htbp]
   \small
   \vspace{-0.8cm}
        \centering	
         \includegraphics[width=0.4\textwidth]{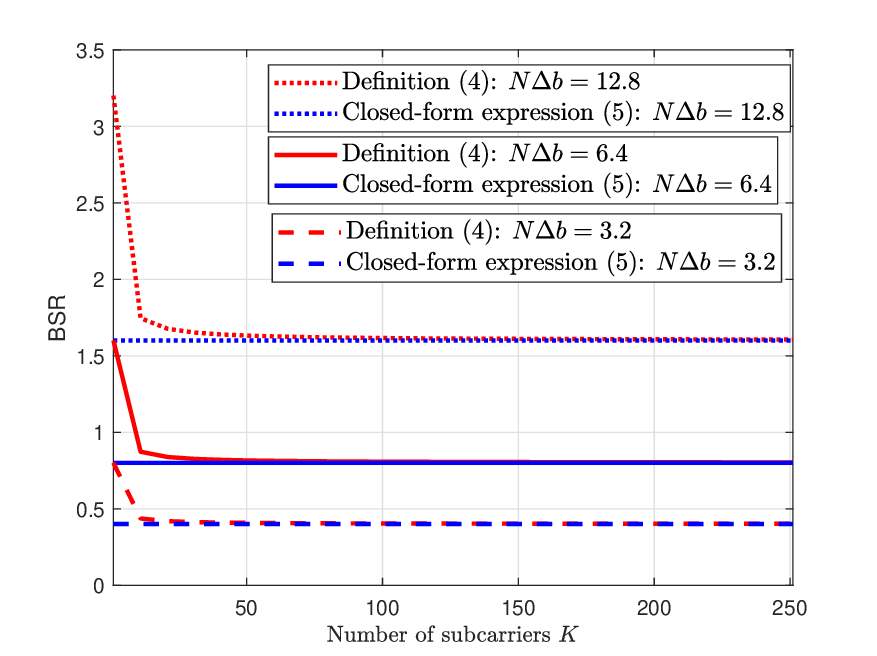}
         \captionsetup{font={small}}
         \vspace{-0.2cm}
        \caption{The BSR obtained by the definition \eqref{eq:BSR definition} and the closed-form expression \eqref{eq:BSR closed-form} with $N \Delta b=3.2,6.4, 12.8$.}
        	\label{fig:BSR vs number of subcarriers}
        \vspace{-0.4cm}
  \end{figure}
  
  Fig. \ref{fig:BSR vs number of subcarriers} shows the proposed closed-form expression of the BSR  \eqref{eq:BSR closed-form} and its analytical definition \eqref{eq:BSR definition}. It can be observed that the closed form is sufficiently accurate for $K\geq 100$. Since the BSR concept targets systems with a large number of subcarriers, which is the typical case in wideband multicarrier communications, the approximation in \eqref{eq:BSR closed-form} can adequately approach the value given by \eqref{eq:BSR definition}. Proposition \ref{remark:BSR equation} implies that the severity of the beam squint effect \textbf{linearly} increases with the number of antennas $N$, fractional bandwidth $b$, and normalized antenna spacing $\Delta$, which was not observed in \cite{dai2022delay} and aligned with observations in \cite{ozen2023interference}.  Furthermore, it shows that the beam squint effect becomes more severe with a larger \textbf{fractional bandwidth} $b$ (rather than bandwidth) but can be mitigated by reducing the antenna spacing $\Delta$, e.g., via spatial oversampling. %\blue{Such results aligns with observations in \cite{ozen2023interference}. }%For example, for a system with $1~\ghz$ bandwidth and 256-element ULAs of half-wavelength antenna spacing,  we obtain $\BSR=0.57$ for $\fc=28~\ghz$ while $\BSR=0.05$ for $\fc=300~\ghz$.} Reducing antenna spacing $\Delta$, e.g., by spatial oversampling, can effectively mitigate the beam squint effect. 
  This is because an array with a reduced antenna spacing forms a larger beamwidth, enhancing the array gain of the marginal frequencies in the bandwidth. As shown in Figs. \ref{fig:BSR2} and \ref{fig:BSR3} on the top of the next page, when $\Delta$ decreases from $1/2$ to $1/4$, the beamwidth of the main lobes is broadened, and at the same time, the BSR is reduced by half. Fig.\ \ref{fig:BSR_all} shows the normalized array gains versus the subcarrier index. It is seen that the beam squint effect becomes negligible for BSR$= 0.1$. We will further demonstrate this through simulations in Section \ref{sc:simulation results}.

  \begin{figure*}[htbp]
  \small
  \vspace{-0.5cm}
        \centering
        \hspace{-8mm}
        \subfigure[${\rm BSR}= 1,\Delta=1/2,N=160$.]
        { \label{fig:BSR2} \includegraphics[width=0.34\textwidth]{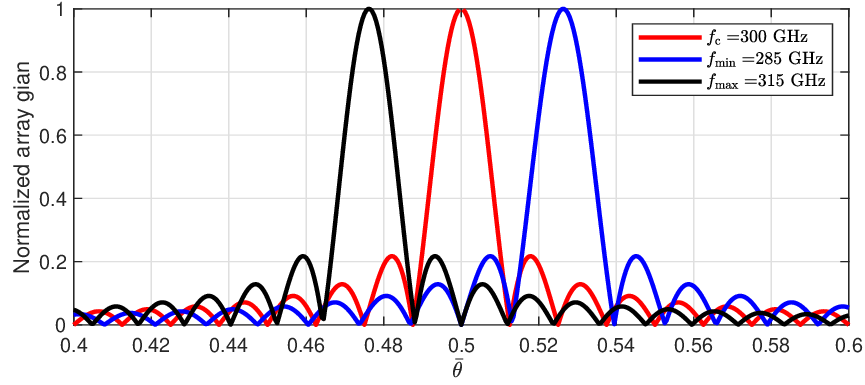}}
         % \hspace{-3mm} 
        \subfigure[${\rm BSR}= 0.5,\Delta=1/4,N=160$. ]
        {\label{fig:BSR3}\includegraphics[width=0.33\textwidth]{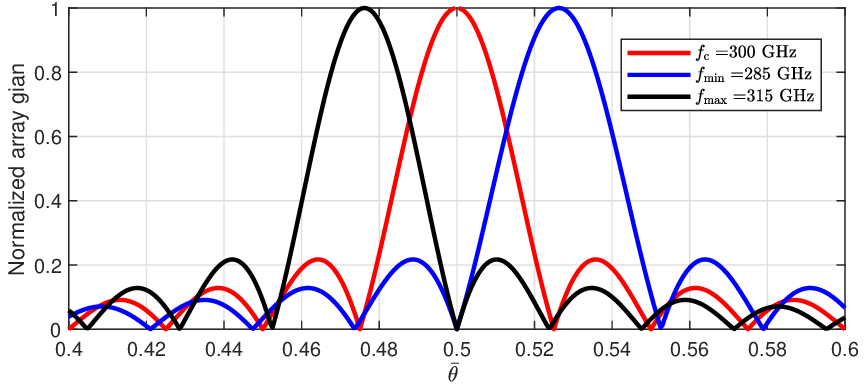}}
         % \hspace{-4mm} 
        \subfigure[Normalized array gain for all subcarriers with $\Delta=1/2$.]
        % \vspace{-0.5cm}
        {\label{fig:BSR_all} \includegraphics[width=0.33\textwidth]{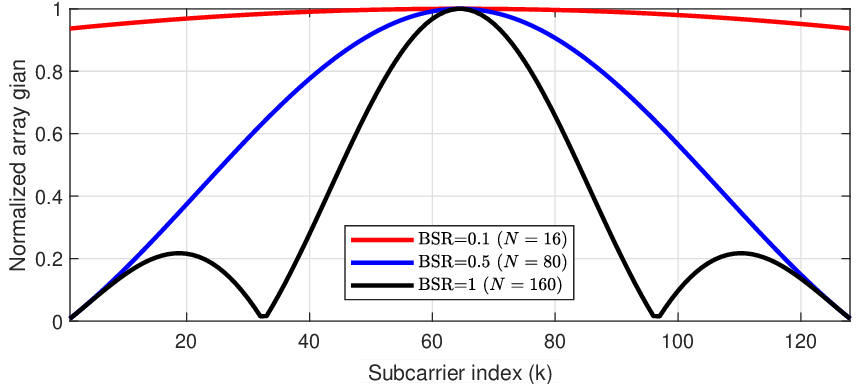}}
        \captionsetup{font={small}}
        \vspace{-0.3cm}
        \caption{ Illustration of beam squint effect in multicarrier systems using PS-based ULAs with $f_{\rm c}=300~{\rm GHz},B=30~{\rm GHz}, K=128$. Figs.\ (a) and (b) present the normalized array gain of three frequencies, namely the central frequency, maximum frequency, and minimum frequency, versus $\vartheta$. Fig.\ (c) indicates the normalized array for all subcarriers in three BSR cases, i.e., $({\rm BSR}= 0.1,N=16)$, $({\rm BSR}= 0.5,,N=80)$ and $({\rm BSR}= 1,N=160)$, respectively.}
        \label{fig:Beam squint effect for different BSR}
        	\vspace{-0.5cm}
  \end{figure*}

% \vspace{-5mm}

\subsection{EAGs of SW- and PS-based Beamformers}\label{sec:EAG of SW and PS array}

\subsubsection{PS-based Beamformer}

The normalized array gain achieved by the PS-based analog beamforming vector $\tilde{\wb}=\left[e^{j\varpi _1}, \ldots,e^{j\varpi_n},\ldots,e^{ j\varpi _N} \right]^T$ for AoA $\theta$ at frequency $f_k$ can be given as \cite{cai2016effect} 
\begin{align}\label{eq:func NAG}
  g\left( \tilde{\wb}, \theta,f_k \right)&=\frac{1}{\sqrt{N}}\left| \tilde{\wb}^H \ab_{\rm r}\left( \theta,f_k \right) \right|  \nonumber \\
  &=\frac{1}{N}\left|\sum\limits_{n=1}^{N}e^{j\left[\varpi_n-2\pi(n-1)\Delta  \frac{f_k}{f_{\rm c}}\sin \theta \right]}\right|.
\end{align}
For a multicarrier system, the beamformer $\tilde{\wb}$ is conventionally designed for the central frequency, i.e., $f_{\rm{c}}$. To obtain the highest array gain, the beamformer $\tilde{\wb}$ should align with the AoA $\theta$ of the arriving signal \cite{cai2016effect}. Specifically, the $n$-th phase shift in $\tilde{\wb}$ should be set as $\varpi_n=2\pi(n-1)\Delta \vartheta, n=1,\ldots, N$.
% \begin{equation}\label{eq:aligned beamformer}
%     \varpi_n=2\pi(n-1)\Delta \sin \theta, n=1,\cdots, N.
% \end{equation}
With this, we can rewrite \eqref{eq:func NAG} as
\begin{align}\label{eq:PS NAG bi-variable}
        g\left(\vartheta , \xi_k \right)&=\frac{1}{N}\left| \sum\limits_{n=1}^{N}e^{j\left[2\pi(n-1)\Delta \left( 1-\xi_k\right) \vartheta   \right]} \right| \nonumber \\
        &=\left|  \frac{\sin\left( N\pi \Delta \left(1- \xi_k\right)\vartheta  \right)}{N \sin\left( \pi \Delta \left(1- \xi_k\right)\vartheta  \right)}\right|.
\end{align}
Then, the EAG over all possible physical directions is
\begin{equation}\label{eq:EAG for PS array}
           \Es_{\rm ps}[g\left( \vartheta,\xi_k \right)]= \frac{1}{2K} \sum\limits_{k=1}^K\int^{1}_{-1}  g\left( \vartheta,\xi_k \right) d \vartheta,
\end{equation}
which can be approximated by a closed-form expression as in the following proposition.
% \vspace{-5mm}
\begin{proposition}\label{lemma:EAG of PS analog beamformer}
The EAG of the PS-based array, i.e., $\Es_{\rm ps}[g\left( \vartheta,\xi_k \right)]$, can be approximated as
\begin{align}\label{eq:approxmated EAG for PS-based array}
   \Es_{\rm ps}[g\left( \vartheta,\xi_k \right)]  \approx \frac{2}{ 3~\text{BSR}} \int_0^{\text{BSR}} \left|{\rm sinc}\left( 4x\right)\right|dx + \frac{1}{3}, 
\end{align}
where ${\rm sinc}(x)\triangleq \frac{\sin(\pi x)}{\pi x}$ and BSR is given as  \eqref{eq:BSR closed-form}. The EAG $\Es_{\rm ps}[g\left( \vartheta,\xi_k \right)]$ monotonically decreases with \text{BSR} and is bounded by $\frac{1}{3}\leq \Es_{\rm ps}[g\left( \vartheta,\xi_k \right)]\leq 1$.
The lower bound is attained as $\text{BSR} \rightarrow \infty$, while the upper bound is achieved as $\text{BSR} \rightarrow 0$.
\end{proposition}
\vspace{-3mm}
 \IEEEproof See Appendix \ref{appd_arraygain_PS}. \qedsymbol

% \vspace{-3mm}
\subsubsection{SW-based Beamformer}\label{sec:SW-based array}

    The normalized array gain achieved by the SW-based analog beamforming vector $\wb = [w_1, \ldots, w_N]^T, w_n \in \left\{0,1\right\}, n = 1, \ldots, N$  is given by   
\begin{align}\label{eq:base matrix for SW-BSR}
  g\left( \wb, \vartheta,\xi_k \right)&=\frac{1}{\sqrt{\left\|\wb \right\|_1}}\left| \wb^H \ab_{\rm r}\left( \theta,f_k \right) \right| \nonumber \\
  &=\frac{1}{\sqrt{N \left\|\wb \right\|_1}}\left|\sum\limits_{n=1}^{N} w_n e^{-j2\pi(n-1)\Delta  \xi_k \vartheta }\right|.
\end{align}
The EAG of the SW-based array can be written as % in \eqref{eq:base matrix for SW-BSR}, i.e.,
\begin{align}\label{eq:EAG of switch bases array}
    \Es_{\rm sw}[g\left( \wb, \vartheta\right)] =  \frac{1}{2K} \sum_{k=1}^K  \int^{1}_{-1} g\left( \wb, \vartheta,\xi_k \right) d\vartheta, 
\end{align}
which is intractable for analysis due to the binary-valued vector $\wb$. Instead, we obtain its approximated closed-form expression detailed in the following proposition.
\begin{proposition}\label{prop:EAG of SW-based array}
    The EAG of the SW-based array, i.e., $\Es_{\rm sw}[g\left( \wb, \vartheta\right)]$, can be approximated as
          \begin{equation}\label{eq:EAG SW expression}
         \Es_{\rm sw}[g\left( \wb, \vartheta\right)] \approx \frac{2}{3}\sqrt{\frac{\left\|\wb \right\|_1}{N}} \leq \frac{2}{3}.
    \end{equation}  
\end{proposition}
 \IEEEproof See Appendix~\ref{appdix:proof of SW EAG}. \qedsymbol

% Based on the above EAG analysis, we have the following remark.
%  \begin{remark}\label{rmk:EAG comparison between PS and SW-based array}
%     The EAG of the PS-based beamformer is significantly affected by the beam squint effect. Specifically, it monotonically decreases with the BSR and approaches $1/3$ when the beam squint effect is relatively severe. In contrast, the EAG of the SW-based beamformer depends on the number of active antennas and is more robust to the beam squint effect. It can be larger than $1/3$ if $\left\|\wb \right\|_1\geq \frac{N}{4}$. Hence, the EAG of the SW-based array can be higher than that of the PS-based array in a wideband system.
% \end{remark}   
  
%   Remark \ref{rmk:EAG comparison between PS and SW-based array} reveals that the SW-HBF may perform better than the PS-HBF in a wideband system. However, we note that directly maximizing the EAG does not necessarily ensure optimal system SE, which motivates us to develop efficient SW-HBF schemes as detailed next. 

With Proposition~\ref{lemma:EAG of PS analog beamformer} and Proposition~\ref{prop:EAG of SW-based array}, we can observe that the EAG of the PS-based beamformer can be smaller than that of the SW-based one when the beam squint effect is severe. For example, when BSR is sufficiently large,  $\Es_{\rm ps}[g\left( \vartheta,\xi_k \right)] \rightarrow \frac{1}{3}$. In contrast, with $\left\|\wb \right\|_1 > \frac{N}{4}$, we obtain $\Es_{\rm sw}[g\left( \wb, \vartheta\right)] >\frac{1}{3}$. This occurs because the PS-based beamformer is highly sensitive to the beam squint effect while the SW-based one is not. Consequently, the relative resilience of the SW-based beamformer to beam squint, compared to the PS-based one, may offer advantages for deploying the SW-HBF architecture over the PS-HBF one in wideband systems. We will further demonstrate this by the numerical results of Section~\ref{sc:simulation results}. %However, we note that directly maximizing the EAG does not necessarily ensure SE performance, which motivates us to develop efficient SW-HBF schemes as detailed next. 

  % \end{bluenote}
  
% \vspace{-2mm}
\section{SW-HBF Design for SU-MIMO}\label{sc:SW-HBF in SU-MIMO} 
In this section, we focus on designing efficient SW-HBF schemes to maximize SE in wideband systems.
 
%  The binary nature of analog beamformer renders the intractable problem of HBF design in wideband multi-carrier systems more challenging. In this section, we propose efficient SW-HBF schemes by first decoupling the design of digital  beamformers and analog beamformer which is solved via advanced search algorithms. Then, the digital beamformers are obtained with closed-form solutions.
% \vspace{-5mm}
\subsection{Problem Formulation}

Considering the Gaussian signalling, based on signal model \eqref{eq:received signal}, the SE of the $k$-th subcarrier, denoted by $R_k$, is expressed as\cite{lin2019transceiver,li2020dynamic}
\begin{equation}\label{eq:SE for each subcarrier}
  R_k=\log_{2}\left|\Ib_{N_{\rm S}}+\frac{1}{\sigma_{\rm n}^{2}} \Wb_k^{\dagger}  \Hb_k \Fb_k \Fb^H_k \Hb^H_k \Wb_k \right|,
\end{equation}
where we recall that $\Fb_k=\Fb_{\rm RF}\Fb_{\rm BB}[k]$ and $\Wb_k=\Wb_{\rm RF}\Wb_{\rm BB}[k]$. %\blue{Compared with the conventional fully digital beamforming and PS-HBF architectures, SW-HBF requires much lower power consumption. Therefore, we investigate its performance by maximizing the average SE under the constraints of transmit power and analog components, as is widely adopted in HBF designs \cite{el2014spatially,ma2021closed,alkhateeb2016frequency,sohrabi2017hybrid,li2020dynamic,yan2022energy}.} 

We aim at a spectral-efficient design of the considered system by optimizing the HBF precoders and combiners to maximize the average SE of the system \cite{el2014spatially, alkhateeb2016frequency, sohrabi2017hybrid}. Let $\left\{ \Fb_k\right\} \triangleq \{ \Fb_1, \ldots, \Fb_K\}$ and $\left\{\Wb_k\right\} \triangleq \{\Wb_1, \ldots, \Wb_K \}$. The problem can be formulated as
% \vspace{-5mm}
\begin{subequations}\label{eq:problem formulation}
  \begin{align}
    \underset{ \{ \Fb_k \}, \{ \Wb_k\}}{\text{max}} \quad  & \frac{1}{K} \sum_{k=1}^{K} R_k\\
    \text{s.t.} \quad  & \| \Fb_k \|^2_F \leq P_{\rm b},\ \forall k, \\
    &\Fb_{\rm RF}\in \Bcl ^{N_{\rm T} \times N_{\rm RF}}, \rank(\Fb_{\rm RF}) \geq N_{\rm S}, \label{eq: analog precoder constraint}\\
    &\Wb_{\rm RF}\in \Bcl ^{N_{\rm R} \times N_{\rm RF}},\rank(\Wb_{\rm RF}) \geq N_{\rm S} \label{eq: analog combiner constraint}.
\end{align}
\end{subequations}
% \vspace{-10mm}
The joint optimization of the transmit and receive beamformers, i.e., $\left\{ \Fb_k \right\}$ and $\left\{ \Wb_k \right\}$, in problem \eqref{eq:problem formulation} is challenging. It is a non-convex non-smooth mix-integer programming problem due to the binary and rank constraints of the analog beamformers. To overcome the challenges, we first decouple the problem into two subproblems: transmit and receive beamforming design, and then solve them successively. The procedure is elaborated below.

% In each subproblem, the analog beamformer is efficiently found by judiciously designed search algorithms. Then the digital beamformers are obtained by closed-form solutions. The details are presented in the following subsections.

\subsection{Transmit HBF Design}\label{sec:Tx HBF Design}
% \vspace{-5mm}
Assuming an optimal solution to $\Wb_k$ \cite{sohrabi2016hybrid,sohrabi2017hybrid}, the analog and digital precoders, i.e., $\Fb_{\rm RF}$ and $\left\{ \Fb_{\rm BB}[k]\right\}$, can be obtained by solving the following problem 
\begin{subequations}\label{eq:problem formulation at Tx}
\begin{align}
        \underset{\Fb_{\rm RF}, \left\{ \Fb_{\rm BB}[k]\right\}}{\text{max}}\quad  & \frac{1}{K} \sum_{k=1}^K \tilde{R}_k\\
        \text {s.t.} \quad & \eqref{eq: analog precoder constraint}, \nonumber \\
        & \| \Fb_{\rm RF}\Fb_{\rm BB}[k] \|^2_F \leq P_{\rm b},\ \forall k,  \label{eq:coupled analog and digital beamformer}
\end{align}
\end{subequations}
where $\tilde{R}_k$ denotes the achievable rate of the transmit precoder, i.e., 
\begin{equation}\label{eq:Rk tilder precoder}
   \tilde{R}_k=\log_{2}\left|\Ib+\frac{1}{\sigma_{\rm n}^{2}} \Hb_k \Fb_{\rm RF}\Fb_{\rm BB}[k] \Fb_{\rm BB}^H[k]\Fb_{\rm RF}^H \Hb^H_k \right|. 
\end{equation} 
% To tackle the highly coupled variables $\Fb_{\rm RF}$ and $\Fb_{\rm BB}[k]$ in \eqref{eq:coupled analog and digital beamformer}, we first introduce auxiliary variables $\hat{\Fb}_{\rm BB}[k]\in \Cs^{N_{\rm RF} \times N_{\rm S}},\forall k $ such that $\Fb_{\rm BB}[k]\triangleq (\Fb_{\rm RF}^H\Fb_{\rm RF})^{-\frac{1}{2}}\hat{\Fb}_{\rm BB}[k] $. Thus, we have $\| \Fb_{\rm RF}\Fb_{\rm BB}[k] \|^2_F = \| \hat{\Fb}_{\rm BB}[k] \|^2_F$, and constraints \eqref{eq:coupled analog and digital beamformer} can be recast as $ \| \hat{\Fb}_{\rm BB}[k] \|^2_F \leq P_{\rm b},\ \forall k$.
% % \begin{equation}\label{eq:Tx power constraints reformulated}
% %     \| \hat{\Fb}_{\rm BB}[k] \|^2_F \leq P_{\rm b},\ \forall k.
% % \end{equation}
% % \begin{subequations}\label{eq:equivalent problem formulation at Tx}
% %   \begin{align}
% %         \max\limits_{\Fb_{\rm RF}, \left\{ \hat{\Fb}_{\rm BB}[k] \right\} } & \frac{1}{K} \sum_{k=1}^{K} R_k\\
% %         \text { s.t. } \quad & \eqref{eq: analog precoder binary constraint}~ \text{and} ~\eqref{eq:precoder rank constr}~ \text{are satisfied}, \nonumber \\
% %       & \| \hat{\Fb}_{\rm BB}[k] \|^2_F \leq P_{\rm b},\ \forall k,
% %     \end{align}
% % \end{subequations}
% Therefore, $\tilde{R}_k$ is re-expressed as
% \begin{equation}\label{eq:equivalent SE for Tx}
%   \tilde{R}_k=\log_{2}\left|\Ib+\frac{1}{\sigma_{\rm n}^{2}} \Hb_{\rm eff}[k]\hat{\Fb}_{\rm BB}[k] \hat{\Fb}_{\rm BB}[k]^H  \Hb_{\rm eff}[k]^H  \right| ,
% \end{equation}
% with $\Hb_{\rm eff}[k]\triangleq \Hb_k \Fb_{\rm RF}(\Fb_{\rm RF}^H\Fb_{\rm RF})^{-\frac{1}{2}}$. 
The binary and rank constraints in \eqref{eq: analog precoder constraint} make problem \eqref{eq:problem formulation at Tx} challenging non-convex mixed-integer, and it is difficult to jointly and optimally solve $\{ \Fb_{\rm RF}, \left\{ \Fb_{\rm BB}[k]\right\} \}$. We consider their decoupled designs in the following.% the designs of the analog and digital precoders next.

% \cite{sohrabi2017hybrid,chen2020hybrid}.
\subsubsection{Digital Precoding Design}
 Define $\Hb_{\rm eff}[k]\triangleq \Hb_k \Fb_{\rm RF}(\Fb_{\rm RF}^H\Fb_{\rm RF})^{-\frac{1}{2}}$. With a given analog precoder, the optimal digital precoder for the $k$-th subcarrier is given as \cite{alkhateeb2016frequency}
\begin{equation}\label{eq:digital precoder}
  \Fb_{\rm BB}[k]=(\Fb_{\rm RF}^H\Fb_{\rm RF})^{-\frac{1}{2}}\Vb_k\Gammab_k^{\frac{1}{2}},
\end{equation}
where $\Vb_k\in \Cs^{N_{\rm RF}\times N_{\rm S}}$ is the matrix whose columns are the $N_{\rm S}$ right-singular vectors corresponding to the $N_{\rm S}$ largest singular values of $\Hb_{\rm eff}[k]$, and $\Gammab_k \in \Cs^{N_{\rm S}\times N_{\rm S}}$ is the diagonal matrix whose diagonal elements are the power allocated to the $N_{\rm S}$ data streams at the $k$-th subcarrier, with $\tr(\Gammab_k)=P_{\rm b}$.

\subsubsection{Analog Precoding Design}\label{sec: analog precoding SU}

% Even with given digital precoders, the design of the analog one is still significantly challenging because of the binary and rank constraints. 

Inserting \eqref{eq:digital precoder} into \eqref{eq:Rk tilder precoder}, we obtain an intractable function of $\Fb_{\rm RF}$. Alternatively, we can optimize its tight upper bound considering \cite{sohrabi2017hybrid}
\begin{equation}\label{eq:SE upper bound}
    \tilde{R}_k \leq \log_{2}\left| \Ib + \frac{\gamma}{\sigma_{\rm n}^2} \Fb_{\rm RF}^{\dagger} \Hb_k^H \Hb_k \Fb_{\rm RF}\right|,
\end{equation}
where $\gamma=P_{\rm b}/N_{\rm S}$ represents the power allocated to each data stream, and the relaxation is tight when $N_{\rm S}=N_{\rm RF}$. We note that for moderate and high signal-to-noise-ratio (SNR) regimes, it is asymptotically optimal to adopt the equal power allocation for all streams in each subcarrier \cite{lee2008downlink,sohrabi2017hybrid}.
By defining $\Hb_{\rm t}[k]\triangleq \Hb_k^H \Hb_k$, the problem of designing the analog precoder $\Fb_{\rm RF}$ is reformulated as
\begin{align}\label{eq:analog precoder design}
     \underset{\Fb_{\rm RF}}{\text{max}} \quad & f\left( \Fb_{\rm RF}\right) \triangleq  \frac{1}{K} \sum_{k=1}^{K} \log_{2}\left| \Ib + \frac{\gamma}{\sigma_{\rm n}^2} \Fb_{\rm RF}^{\dagger} \Hb_{\rm t}[k] \Fb_{\rm RF}\right|\\[-10pt]
        \text {s.t.} \quad & \eqref{eq: analog precoder constraint}. \nonumber 
\end{align}
% \vspace{-1mm}
Problem \eqref{eq:analog precoder design} remains challenging due to the binary and rank constraint of $\Fb_{\rm RF}$, and there is no existing solution to it in the literature. A relaxation of the binary variables may cause significant performance loss. Although the optimal solution can be obtained via the ES method, the resultant algorithm requires prohibitively high complexity. For such challenging mixed-integer programming, the tabu search (TS) approach \cite{laguna2018tabu, nguyen2019qr} can find a near-optimal solution with reduced complexity compared to the ES method. Specifically, the TS procedure explores candidates and their neighbors within the feasible region over iterations. In each iteration, the best non-tabu (i.e., unforbidden) neighbor is chosen as the candidate and stored in the tabu list to avoid cycling. As such, the TS method ensures convergence to a local optimum with a sufficiently large number of iterations \cite{laguna2018tabu}. In TS-based methods, the design of the neighbor set plays a pivotal role in balance performance and complexity. A larger neighbor set can generally enable better performance but result in a higher complexity. To efficiently tackle problem \eqref{eq:analog precoder design}, we introduce a PGA-TS-aided analog precoding algorithm. This approach effectively diminishes the size of the neighbor set while maintaining satisfactory performance.  

 % However, the search procedure generally requires high complexity, especially in large-scale MIMO systems. To overcome this, we propose a PGA-TS-aided analog precoding algorithm next.  

The idea of the PGA-TS scheme is to first apply the PGA method to find a relaxed but efficient solution in $[0,1]$. Entries that are neither $0$ nor $1$ in the relaxed solution are considered erroneous and require further optimization with the TS procedure. By doing so, the search space of the TS procedure can be significantly diminished, facilitating faster convergence and lower complexity. Specifically, we first relax the rank and binary constraints in problem \eqref{eq:analog precoder design} and consider the following problem
% \vspace{-3mm}
\begin{equation}\label{eq:PGA problem}
 \setlength{\abovedisplayskip}{0pt}
 \setlength{\belowdisplayskip}{0pt}
          \underset{\Fb_{\rm RF}}{\text{max}} \quad f\left( \Fb_{\rm RF}\right),\ \text{s.t.} \quad \Fb_{\rm RF}\in \Fcl,
\end{equation}
where $\Fcl\triangleq\{\Fb_{\rm RF}| \Fb_{\rm RF}(i,j)\in [0,1], \forall i,j \} $.
Problem \eqref{eq:PGA problem} can be efficiently solved using the PGA method, as summarized in Algorithm \ref{alg:PGA algorithm}, where $\nabla_{\Fb_{\rm RF}}f(\Fb_{\rm RF})$ denotes the gradient of $f(\Fb_{\rm RF})$, given by \cite{petersen2008matrix}
  \begin{align}\label{eq:PGA gradient P2P}
  \nabla_{\Fb_{\rm RF}}f(\Fb_{\rm RF})&=\frac{2}{K} \sum_{k=1}^{K} \Ab_k\Fb_{\rm RF}\left(\Fb_{\rm RF}^H\Ab_k\Fb_{\rm RF}\right)^{-1} \nonumber \\
    & \qquad \qquad -2\Fb_{\rm RF}\left(\Fb_{\rm RF}^H\Fb_{\rm RF}\right)^{-1}
\end{align}
with $\Ab_k\triangleq\Ib+\frac{\gamma}{\sigma_{\rm n}^2} \Hb_{\rm t}[k] $. In step 4, the normalized gradient $\nabla \Frf$ is obtained and employed later in step 6, where the step size $\mu$ is determined by the backtracking line search method \cite{bertsekas1997nonlinear} as in step 5, and $[\cdot]_{\Fcl}$ denotes the projection into $\Fcl$. The iterations continue until the procedure satisfies the stopping criteria. Since the objective function is non-decreasing during the iteration, the PGA algorithm is guaranteed to converge to a stationary point of the non-convex problem \eqref{eq:PGA problem}.
% \vspace{-0.5cm}

\begin{algorithm}[t]
 \small
% \setstretch{1}
\caption{ PGA algorithm for problem \eqref{eq:PGA problem}}\label{alg:PGA algorithm}
\LinesNumbered %要求显示行号
% \KwIn{$i=0,\Fb_{\rm RF}^{(i)}\in \Fc,c$}
\KwOut{$\Fb_{\rm RF}^{\rm pga}$}
Initialize $\alpha \in (0,0.5)$, $ \beta \in (0,1)$, $\Fb_{\rm RF}\in \Fcl$, and $\epsilon$\;
\Repeat{$  |f(\Frf')-f(\Frf)| \leq \epsilon$ }{
Cache the solution $\Frf'\leftarrow \Frf$

  Obtain $\nabla  \Frf =\frac{\nabla_{\Fb_{\rm RF}}f(\Fb_{\rm RF})}{\left\| \nabla_{\Fb_{\rm RF}}f(\Fb_{\rm RF})\right\|_F}$ with \eqref{eq:PGA gradient P2P}.

 Start with $\mu=1$, repeat $\mu \leftarrow \beta \mu$ until 
 $f\left([\Fb_{\rm RF}+\mu\nabla  \Frf]_{\Fcl}\right)> f\left( \Frf\right)+ \alpha \mu \|\nabla \Frf\|_F^2$.
  %Chose step size $\mu$ via backtracking line search with $\alpha$ and $\beta$.
 
 Update $\Frf$ as $\Fb_{\rm RF} \leftarrow [\Fb_{\rm RF}+\mu\nabla  \Frf]_{\Fcl}$.
 }
 Obtain $\Fb_{\rm RF}^{\rm pga}$ as $\Fb_{\rm RF}^{\rm pga}= \Fb_{\rm RF}$.
\end{algorithm}

% \vspace{-3mm}

Observing that the non-integer entries of $\Fb_{\rm RF}^{\rm pga}$ fall into $(0,1)$, we further refine the solution before performing the TS procedure for efficiency. Specifically, define the following refinement function
\begin{equation}\label{eq:definition of refinement function}
   r\left(x\right)\triangleq \begin{cases}
\delta & \text{ if } \left|x-\delta \right|\leq \varepsilon,  \\
 x & {\rm otherwise},
\end{cases},
\end{equation}
where $\delta\in \{0,1\}$ and $\varepsilon>0$ is a threshold rounding an entry to a binary value with an error upper bound $\varepsilon$. Applying $r(\cdot)$ to the entries of $\Fb_{\rm RF}^{\rm pga}$, we obtain a refined solution $\tilde{\Fb}_{\rm RF}^{\rm pga}$, i.e., $\tilde{\Fb}_{\rm RF}^{\rm pga}(i,j)=r\left(\Fb_{\rm RF}^{\rm pga}(i,j)\right), \forall i,j$, in which the entries close to $\delta\in \{0,1\}$ are rounded to $\delta$.% Thus, the number of non-integer elements of $\tilde{\Fb}_{\rm RF}^{\rm pga}$ is smaller than that of $\Fb_{\rm RF}^{\rm pga}$, which further reduces the search space of the TS algorithm. 

The PGA-aided TS (PGA-TS) algorithm is summarized in Algorithm \ref{alg:APGA-TS algorithm}. In step 1, the initial solution $\tilde{\Fb}_{\rm RF}^{\rm pga}$ is obtained by applying the PGA algorithm and the refinement $r(\cdot)$, as explained above. To facilitate the recovery of the analog precoder, the indices of non-integer entries of $\tilde{\Fb}_{\rm RF}^{\rm pga}$ are stored in the set $\Scl_I$. Then, in the remaining steps, the TS procedure is performed in the reduced search space $\Bcl^{|\Scl_I|}$, where we recall that $\Bcl = \{0, 1\}$. Specifically, let $\qb^{(i)} \in \Bcl^{|\Scl_I|}$ be the candidate of the TS procedure in the $i$-th iteration. The neighbor set $\Ncl(\qb^{(i)})$ is defined as the set consisting of $|\Scl_I|$-dimension vectors that have only one element different from $\qb^{(i)}$. Let $\Nnb \triangleq |\Ncl(\qb^{(i)})|$ denote the size of the neighbor set satisfying $\Nnb \leq |\Scl_I| \leq \Nt\Nrf$. Note that because only the analog precoder entries with indices in $\Scl_I$ are updated in each iteration, the overall analog precoder needs to be recovered as in step 5. Specifically, based on $\Scl_I$, the entries of $\qb^{(i)}$ can be inserted back to the corresponding non-integer position in $\tilde{\Fb}_{\rm RF}^{\rm pga}$ to recover the analog precoder. Here, an analog precoder is only valid if it satisfies the rank constraint \eqref{eq: analog precoder constraint}. Then in step 6, the neighbor set is examined to find the best neighbor $\tilde{\qb}^{(i)\star}$ whose corresponding analog precoder $\tilde{\Fb}_{\rm RF}^{(i)\star}$ yields the largest objective value. Note that the best neighbor has to be not in the tabu list. The best solution is updated if a larger objective value is found, as in steps 7--9. The best neighbor $\tilde{\qb}^{(i)\star}$ is then pushed into the tabu list in step 10 to avoid cycles. Finally, the candidate is updated for the next iteration, as in step 11. This iterative procedure can be terminated once the objective value converges or when the number of iterations exceeds a predefined threshold.

 %During the iterations, a tabu list is maintained to record visited solutions to avoid cycles and guarantee convergence. % which can be significantly smaller than $\Bc^{N_{\rm T}N_{\rm RF}}$. Therefore, the PGA-TS algorithm can find the solution to the analog precoder with reduced complexity and faster convergence. 

The convergence of Algorithm \ref{alg:APGA-TS algorithm} is guaranteed because the objective values form a nondecreasing sequence over iterations. Furthermore, with the aid of PGA solution and refinement, the size of the neighbor set has decreased from $\Nt\Nr$ to $|\Scl_I|$, which could result in $2^{\frac{\Nt\Nrf}{2}}$ times reduction in the size of the search space based on our numerical simulations. Meanwhile, the PGA solution facilitates the search for a satisfactory local optimum, thereby significantly improving the efficiency of the iterative TS procedure, especially for the considered large-scale MIMO systems.
 
 % \setlength{\textfloatsep}{0.1cm}
  % \vspace{-5mm}
 \begin{algorithm}[t]
 \small
% \setstretch{1}
\caption{PGA-TS for designing analog precoder}\label{alg:APGA-TS algorithm}
\LinesNumbered %要求显示行号
%\KwIn{$i=0,L_{\rm ts},\Fb^i,N_{\rm iter}$}
\KwOut{$\Fb_{\rm RF}^{\star}$}
Perform the PGA algorithm to obtain $\Fb_{\rm RF}^{\rm pga}$, and
apply the refinement function $r(\cdot)$ (where $\varepsilon=0.1$) to obtain $\tilde{\Fb}_{\rm RF}^{\rm pga}$. Store the index of those non-integer entries of $\tilde{\Fb}_{\rm RF}^{\rm pga}$ in $\Scl_I$.

Set $i \gets 0$ and generate the initial candidate $\qb^{(0)} \in \Bcl^{|\Scl_I|}$, then recover the analog precoder $\Fb_{\rm RF}^{(0)}$ based on $\qb^{(0)}$ and $\Scl_I$. Set $\qb^{\star} \gets \qb^{(0)}, \Fb_{\rm RF}^{\star} \gets \Fb_{\rm RF}^{(0)} $, $\Lcl_{\rm ts} \gets \Lcl_{\rm ts} \cup \qb^{(0)}$, and $\Nnb$.

\While{$i \leq N_{\rm iter}^{\rm max}\; \textbf{and}$  {\rm not converge}}{
Find neighbor set $\Ncl(\qb^{(i)})$ of the candidate $\qb^{(i)}$.

Recover the analog precoder of each neighbor in $\Ncl(\qb^{(i)})$ with the aid of $\Scl_I$, and cache the analog precoder into $\Fcl_{\qb}^{(i)}$ if it satisfies the rank constraint \eqref{eq: analog precoder constraint}.

Find the best neighbor $\tilde{\qb}^{(i)\star}$ in $\Ncl(\qb^{(i)})$ but not in the tabu list $\Lcl_{\rm ts}$ . The according analog precoder of the best neighbor $\tilde{\qb}^{(i)\star}$, denoted as $\tilde{\Fb}_{\rm RF}^{(i)\star}$, achieves the largest objective value $f(\tilde{\Fb}_{\rm RF}^{(i)\star})$ among $\Fcl_{\qb}^{(i)}$.

    \If{$ f(\tilde{\Fb}_{\rm RF}^{(i)\star}) > f(\Fb_{\rm{RF}}^{\star}) $ }
    {Update the best solution: $\Fb_{\rm{RF}}^{\star} \gets \tilde{\Fb}_{\rm RF}^{(i)\star}$.}
    
        Push $\tilde{\qb}^{(i)\star}$ to the tabu list: $\Lcl_{\rm ts} \gets \Lcl_{\rm ts} \cup \tilde{\qb}^{(i)\star}$, and remove the first candidate in $\Lcl_{\rm ts}$ if it is full.
    
    Set $\tilde{\qb}^{(i)\star}$ as the candidate for the next iteration, i.e., $\qb^{(i+1)} \gets \tilde{\qb}^{(i)\star}$.
    
    $i \gets i+1$.
}
\end{algorithm}
\subsection{Receive HBF Design}\label{sec:Rx HBF Design}
For a given analog combiner $\Wb_{\rm RF}$, the optimal digital combiner of each subcarrier is the MMSE solution \cite{el2014spatially,sohrabi2017hybrid}
% \vspace{-3mm}
 \begin{equation}\label{eq:optimal baseband combiner}
   \Wb_{\rm BB}[k]=\left(\Jb_k\Jb^H_k+\sigma_{\rm n}^2\Wb_{\rm RF}^H\Wb_{\rm RF}\right)^{-1}\Jb_k,
 \end{equation}
where $\Jb_k\triangleq \Wb_{\rm RF}^H \Hb_k\Fb_k$. With \eqref{eq:optimal baseband combiner}, the problem of analog combiner design can be formulated as %\cite{sohrabi2017hybrid}
% \vspace{-3mm}
 \begin{align}\label{eq:analog combiner problem}
    \underset{\Wb_{\rm RF}}{\text{max}} & \quad \frac{1}{K}\sum_{k=1}^{K}\log_2\left|\Ib+\frac{1}{\sigma_{\rm n}^2}\Wb_{\rm RF}^{\dagger}\Tb_k\Wb_{\rm RF}\right| \\[-2pt]
        \text {s.t.}  &\quad \eqref{eq: analog combiner constraint}, \nonumber 
 \end{align}
 where $\Tb_k\triangleq \Hb_k\Fb_k\Fb^H_k\Hb^H_k$. It is observed that the problems of analog precoding and combining designs, i.e., problems \eqref{eq:analog precoder design} and \eqref{eq:analog combiner problem}, have a similar mathematical structure. Therefore, Algorithm \ref{alg:APGA-TS algorithm} can be utilized to solve problem \eqref{eq:analog combiner problem}. We omit to apply the algorithm for the analog combiner design due to the limited space. Instead, we summarize the overall proposed SW-HBF design for both the Tx and Rx sides and the complexity analysis next.
 % \vspace{-4mm}
 \subsection{Overall SW-HBF Design and Complexity Analysis}
\subsubsection{Overall SW-HBF Design} Our proposed SW-HBF scheme is summarized in Algorithm \ref{alg:proposed SW-HBF scheme}, in which, $\Fb_{\rm RF},\{\Fb_{\rm BB}[k]\}$ and $\Wb_{\rm RF},\{\Wb_{\rm BB}[k]\}$ are solved sequentially. Here, while the analog beamforming matrices $\Fb_{\rm RF}$ and $\Wb_{\rm RF}$ are obtained with Algorithm \ref{alg:APGA-TS algorithm}, the digital ones $\{\Fb_{\rm BB}[k]\}$ and $\{\Wb_{\rm BB}[k]\}$ admit closed-form solutions, as shown in \eqref{eq:digital precoder} and \eqref{eq:optimal baseband combiner}. 
% \vspace{-5mm}

 \begin{algorithm}[hbpt]
\small
% \setstretch{1}
\LinesNumbered %要求显示行号
% \KwIn{$\Hb_k,\forall k$,$\gamma=\frac{P_{\rm b}}{N_{\rm S}}, \sigma^2_{\rm n}$}
\KwOut{$\Fb_{\rm RF},\Wb_{\rm RF} $, $\Fb_{\rm BB}[k],\Wb_{\rm BB}[k], \forall k$}

Obtain $\Fb_{\rm RF}$ by solving problem \eqref{eq:analog precoder design} with Algorithm \ref{alg:APGA-TS algorithm}.

Obtain $\Fb_{\rm BB}[k],\; \forall k$  by \eqref{eq:digital precoder}.

Obtain $\Wb_{\rm RF}$ by solving problem \eqref{eq:analog combiner problem} with modified Algorithm  \ref{alg:APGA-TS algorithm}.

Obtain $\Wb_{\rm BB}[k],\; \forall k$  by \eqref{eq:optimal baseband combiner}.
\caption{Proposed SW-HBF scheme}\label{alg:proposed SW-HBF scheme}

\end{algorithm}

% \vspace{-0.3cm}
\subsubsection{Complexity Analysis}\label{sec:complexity analysis MIMO}
We first evaluate the complexity of obtaining the analog beamformer with Algorithm \ref{alg:APGA-TS algorithm}. Considering $N_{\rm RF} \ll N_{\rm T}$, the complexity to compute the objective function $ f\left( \cdot \right) $ is given as $\Ocl\left( KN_{\rm RF}\left(2N_{\rm T}^2+4N_{\rm T} N_{\rm RF}+N_{\rm RF}^2\right) \right) \approx \Ocl\left( 2KN_{\rm RF}N_{\rm T}^2\right)$, which are mainly caused by Moore–Penrose inversion and matrix multiplications. In Algorithm \ref{alg:APGA-TS algorithm}, the complexity of the PGA algorithm is estimated as $\Ocl\left(N_{\rm pga}KN_{\rm RF}\left(4N_{\rm T}^2+4N_{\rm T} N_{\rm RF}+N_{\rm RF}^2\right) \right) \approx  \Ocl\left(4N_{\rm pga}KN_{\rm RF}N_{\rm T}^2 \right)$, where $N_{\rm pga}$ denotes the number of iterations. Considering that Algorithm \ref{alg:APGA-TS algorithm} requires to compute the objective function $\Nnb$ times in each iteration, the complexity of Algorithm \ref{alg:APGA-TS algorithm} is given by $\Ocl\left(c KN_{\rm RF}N_{\rm T}^2\right)$, where $c\triangleq 4N_{\rm pga}+2N_{\rm iter}^{\rm max}\Nnb$.
Note that $\Nnb$ satisfying $\Nnb \leq |\Scl_I| \leq \Nt\Nrf$ can be predefined to tradeoff the complexity and performance of Algorithm \ref{alg:APGA-TS algorithm}.
 By replacing $N_{\rm T}$ with $N_{\rm R}$, the complexity of obtaining the analog combiner $\Wb_{\rm RF}$ can be derived similarly. In addition, the complexity of steps 2 and 4 of Algorithm \ref{alg:proposed SW-HBF scheme} is approximate $\Ocl\left(2KN_{\rm T}N_{\rm R}\left(N_{\rm T} +N_{\rm R}\right) \right)$, which is mainly due to matrix multiplications. Therefore, the overall complexity of Algorithm \ref{alg:proposed SW-HBF scheme} is estimated as $\Ocl\left(c KN_{\rm RF}N_{\rm T}^2\right)
+\Ocl\left(c KN_{\rm RF}N_{\rm R}^2\right)
+ \Ocl\left(2KN_{\rm T}N_{\rm R}\left(N_{\rm T} +N_{\rm R}\right) \right) =\Ocl\left((c\Nrf+2\Nt)K\Nr^2+ (c\Nrf+2\Nr)K\Nt^2\right) $, where we assume that $N_{\rm pga}$, $N_{\rm iter}^{\rm max}$, and $\Nnb$ in the precoder design are the same as those in the combiner design. In contrast, the complexity of an ES method is approximately $2^{\Nt\Nrf}\Ocl\left(KN_{\rm RF}N_{\rm T}^2\right)+2^{\Nr\Nrf}\Ocl\left(KN_{\rm RF}N_{\rm R}^2\right) $. Hence, the proposed method has a relatively very low complexity. We note that most of the complexity of the proposed method comes from the computations of the objective values \( f(\cdot) \) associated with the neighbor points to find the best one. These can be performed in parallel to reduce runtime. To further reduce the complexity, small $\Nnb$ can be used, and early termination can be employed during the neighbor search. However, these approaches cause performance loss. We will demonstrate this in Section~\ref{sec:simulations for SU-MIMO}.  

% \vspace{-5mm}
\section{Hybrid Precoding Design for MU-MISO}\label{sec:SW-HBF for MU-MISO}
We now consider the SW-based hybrid precoding design in a MU-MISO OFDM system where a base station (BS) with $\Nt$ antennas and $\Nrf$ RF chains serves $\Nu$ single-antenna users. The transmitted signal vector at the $k$-th subcarrier can be written as
\begin{equation}
    \xb[k]=\sum\limits_{i=1}^{\Nu} \Frf \fb_{{\rm BB}_i} [k] s_i[k],
\end{equation}
where $\Frf \in \Bcl^{\Nt\times \Nrf}$ is the analog precoder, $\fb_{{\rm BB}_i} [k] \in \Cs^{\Nrf}$ and $s_i[k]$ are the digital precoder and data symbol intended for user $i$ at subcarrier $k$, respectively, with $\Es[|s_i[k]|^2]=1,\forall i,k$. Let $\hb_m \in \Cs^{\Nt}$ be the channel between the BS and the user $m$. The $m$-th user receives $y_m[k]=\hb_m^H \xb[k]+z_m[k]$, where $z_m[k]$ represents the AWGN. The achievable rate of the user $m$ at subcarrier $k$ can be expressed as
\begin{equation}\label{eq:sum rate expression in MISO}
    R_m[k]=\log_2\left(1+\frac{\left|\hb_m^H[k]\Frf \fb_{{\rm BB}_m}[k] \right|^2}{\sum\limits_{i\neq m}^{\Nu}\left| \hb_m^H[k]\Frf \fb_{{\rm BB}_i}[k] \right|^2 + \sigma_{\rm n}^2} \right),
\end{equation}
where $\sum\limits_{i\neq m}^{\Nu}\left| \hb_m^H[k]\Frf \fb_{{\rm BB}_i}[k] \right|^2$ represents the interference caused by the data intended for other users. Define $\{\Fbb[k]\}\triangleq \{\Fbb[k], k=1,\ldots, K\}$ where $\Fbb[k] \triangleq [\fb_{{\rm BB}_1} [k],\ldots,\fb_{{\rm BB}_{\Nu}} [k] ]$. The hybrid precoding design is formulated as the following average weighted sum rate maximization problem: 
\begin{subequations}\label{pb:hybrid precoding in Multiuser cases}
  \begin{align}
    \underset{ \Frf, \{ \Fbb [k]\}}{\text{max}} \quad  & \frac{1}{K} \sum_{k=1}^{K} \sum_{m=1}^{\Nu} w_m R_m[k]\\
    \text{s.t.} \quad  & \| \Frf \Fbb[k] \|^2_F \leq P_{\rm b},\ \forall k,\label{eq: power constraint multiuser cases} \\
    &\Fb_{\rm RF}\in \Bcl ^{N_{\rm T} \times N_{\rm RF}}, \rank(\Fb_{\rm RF}) \geq \Nu, \label{eq: analog precoder constraint multiuser cases}
\end{align}
\end{subequations}
where $w_m$ is the weight associated with the $m$-th user. The non-convex and mixed-integer nature of problem \eqref{pb:hybrid precoding in Multiuser cases} renders it challenging to solve. Moreover, the interference term involves repeatedly computation for all $\Nu$ users over all $K$ subcarriers, which further complicates the problem. To overcome these challenges, we propose a two-step method, where we first jointly design the baseband precoders and the relaxed analog beamforming matrix, and then obtain the analog beamforming matrix based on the relaxed one. Specifically, in the first step, we solve problem
\begin{subequations}\label{pb:relaxed HBF design in MU}
  \begin{align}
    \underset{ \Frf, \{ \Fbb [k]\}}{\text{max}} \quad  & \frac{1}{K} \sum_{k=1}^{K} \sum_{m=1}^{\Nu} w_m R_m[k]\\
    \text{s.t.} \quad  & \eqref{eq: power constraint multiuser cases}, \nonumber \\
    & \Frf \in \Fcl, \label{eq:relaxed analog BF constraint MU}
\end{align}
\end{subequations}
which is relaxed from \eqref{pb:hybrid precoding in Multiuser cases}. Here, we recall that $\Fcl\triangleq \{\Fb_{\rm RF}| \Fb_{\rm RF}(i,j)\in [0,1], \forall i,j \} $. Then, in the second step, we leverage Algorithm \ref{alg:APGA-TS algorithm} to obtain $\Frf$ satisfying \eqref{eq: analog precoder constraint multiuser cases} based on the solution to problem \eqref{pb:relaxed HBF design in MU}. %We elaborate on the details next.
% \subsection{Step 1}

Despite the relaxation, problem \eqref{pb:relaxed HBF design in MU} is still non-convex, and the analog and digital precoders are still highly coupled. We tackle this by resorting to the fractional programming (FP) theory \cite{shen2018fractional} and transforming it into an equivalent but more tractable expression with the following proposition.

\begin{proposition}
\label{prop:FP equivalent expression}
The objective function is equivalent to the maximization of
\begin{align}\label{eq:fq}
    \hspace{-10mm}&f_q\left(\Frf, \{ \Fbb [k], \rb[k], \qb[k]\}_{k=1}^K\right)= \frac{G}{K\ln2}\nonumber \\
    &+\frac{1}{K} \sum_{(k,m)}  w_m \log_2(1+r_m[k]) -\frac{1}{K\ln 2} \sum_{(k,m)}w_m r_m[k]
\end{align}
% \begin{equation}
%     \max\limits_{\{ \rb[k], \qb[k]\}_{k=1}^K} \quad f_q\left(\Frf, \{ \Fbb [k], \rb[k], \qb[k]\}_{k=1}^K\right)
% \end{equation}
with respect to $\{ \rb[k], \qb[k]\}_{k=1}^K$, where $\rb[k]\in \Rs^{\Nu},\forall k$ and $\qb[k] \in \Cs^{\Nu}, \forall k$ are introduced auxiliary variables, and %$f_q(\Frf, \{ \Fbb [k], \rb[k], \qb[k]\}_{k=1}^K)$ is give as
%
%wi   
\begingroup
\allowdisplaybreaks
\begin{align}\label{eq:G}
   &G \triangleq \sum_{(k,m)}2\sqrt{w_m(r_m[k]+1)}\Re\{q_m[k]^* \hb_m^H[k]\Frf \fb_{{\rm BB}_m}[k]\} \nonumber \\[-5pt]
   & \quad -\sum_{(k,m)} |q_m[k]|^2 \left(\sum\limits_{i} \left| \hb_m^H[k]\Frf \fb_{{\rm BB}_i}[k] \right|^2 + \sigma_{\rm n}^2 \right).
\end{align}
\endgroup
Here, $r_m[k]$ and $q_m[k]$ are the $m$-th elements of $\rb[k]$ and $\qb[k]$, respectively. The optimal solution to $r_m[k]$ and $q_m[k]$ are given as
   \begingroup
\allowdisplaybreaks
    \begin{align}
        & r_m^{\star} [k] = \frac{\left|\hb_m^H[k]\Frf \fb_{{\rm BB}_m}[k] \right|^2}{\sum\limits_{i\neq m}^{\Nu}\left| \hb_m^H[k]\Frf \fb_{{\rm BB}_i}[k] \right|^2 + \sigma_{\rm n}^2} \label{eq:optimal rmk},\\[-5pt]
        &q_m^{\star}[k]=\frac{\sqrt{w_m(r_m[k]+1)}\hb_m^H[k]\Frf \fb_{{\rm BB}_m}[k]}{\sum\limits_{i} \left| \hb_m^H[k]\Frf \fb_{{\rm BB}_i}[k] \right|^2 + \sigma_{\rm n}^2}. \label{eq:optimal qmk}
    \end{align}
\endgroup
\end{proposition}
\vspace{-2mm}
 \IEEEproof See Appendix \ref{sec:appendix B}. \qedsymbol

Proposition \ref{prop:FP equivalent expression} shows that the solutions of \eqref{pb:relaxed HBF design in MU} can be obtained as
\begin{subequations}\label{pb:alternative form in Multiuser cases}
  \begin{align}
   & \{ \Frf^{\star}, \{ \Fbb^{\star}[k],\rb^{\star}[k], \qb^{\star}[k]\}_{k=1}^K\}=\nonumber \\
    &\qquad \arg\max \quad   f_q(\Frf, \{ \Fbb [k], \rb[k], \qb[k]\}_{k=1}^K)\\
    &\qquad \qquad \text{s.t.} \quad \eqref{eq: power constraint multiuser cases}~\text{and}~\eqref{eq:relaxed analog BF constraint MU} \nonumber.
\end{align}
\end{subequations}
Define $\{\rb[k]\}\triangleq \{\rb[k],\; k=1,\ldots, K \}$ and $\{\qb[k]\}\triangleq \{\qb[k],\; k=1,\ldots, K \}$.  It is observed that we can iteratively update $\Frf$, $\{\Fbb[k]\}$, $\{\rb[k]\}$ and $\{\qb[k]\}$ to solve problem \eqref{pb:alternative form in Multiuser cases}. Since the optimal solutions to $\{\rb[k]\}$ and $\{\qb[k]\}$ are given in \eqref{eq:optimal rmk} and \eqref{eq:optimal qmk}, respectively, we aim to develop the solution to $\Frf$ and $\{\Fbb[k]\}$ next. 

\subsection{Analog Precoder Design}
With other variables fixed, the design of the analog preocder $\Frf$ can be expressed as
\begin{equation}\label{pb:analog BF design MU}
    \Frf^{\star} =\arg\max \; g\left( \Fb_{\rm RF}\right) \quad \text{s.t.} \quad \eqref{eq:relaxed analog BF constraint MU},
\end{equation}
with $g\left( \Fb_{\rm RF}\right)\triangleq 2\Re\{\tr \left(\Xib\Frf^H\right)\}-\sum\limits_{k=1}^K \tr\left( \Frf^H\Yb_k\Frf\Zb_k\right)$, where 
\begingroup
\allowdisplaybreaks
\begin{align}
    &\Xib\triangleq \sum_{(k,m)}\sqrt{w_m(r_m[k]+1)}q_m[k] \hb_m[k] \fb_{{\rm BB}_m}[k]^H, \\
    &\Zb_k\triangleq  \sum\limits_{i=1}^{\Nu}\fb_{{\rm BB}_i}[k] \fb_{{\rm BB}_i}[k]^H,\forall k, \\ 
    &\Yb_k\triangleq  \sum\limits_{i=1}^{\Nu} |q_i[k]|^2\hb_i[k]\hb_i^H[k], \forall k. 
\end{align}
\endgroup
 With the gradient of $g(\Frf)$ given as $\nabla_{\Frf} g\left( \Frf\right)=2\Xib-2\sum_{k} \Yb_k\Frf \Zb_k$, we can perform the PGA method similar to Algorithm \ref{alg:PGA algorithm} to efficiently find an optimal solution $\Fb_{\rm RF}^{\star}$ to the convex problem \eqref{pb:analog BF design MU}.
\subsection{Digital Precoder Design}\label{pb:digital precoder design in Multiuser cases}
For given $\Frf$ and $\{\rb[k], \qb[k]\}_{k=1}^K$, the design of the digital precoder is formulated as
\begin{equation}\label{pb:digital BF design MU}
    \{ \Fbb^{\star}[k]\}_{k=1}^K =\arg\max \; G \quad \text{s.t.} \quad \eqref{eq: power constraint multiuser cases},
\end{equation}
where $G$ is given in \eqref{eq:G}.
The optimal solution of problem \eqref{pb:digital BF design MU} can be determined by the Lagrangian multiplier method. Specifically, the Lagrangian function is expressed as
\begin{equation}
    L_G=G+ \sum\limits_{k=1}^K \eta_k\left( P_{\rm b}- \|\Frf \Fbb[k]\|^2_F\right),
\end{equation}
where $\eta_k\geq 0, \forall k$ are introduced Lagrangian multipliers. With the first-order condition of optimality, we can derive the optimal digital precoder of subcarrier $k$ as
% \begin{equation}\label{eq:CF digital BF MU}
%     \fb_{{\rm BB}_i}[k]^{\star}= \Cb_k^{-1}\sqrt{w_i(r_i[k]+1)} \Frf^H \hb_i[k] q_i[k], \forall i,k,
% \end{equation}
% where $\Cb_k\triangleq \Frf^H \left( \Yb_k + \eta_k \Ib\right)\Frf$. With some algebra, we obtain
\begin{equation}\label{eq:CF digital BF compact MU}
    \Fbb[k]^{\star}=\Cb_k^{-1}\Frf^H \Hb[k] \Db_k, \forall k,
\end{equation}
where $\Cb_k\triangleq \Frf^H \left( \Yb_k + \eta_k \Ib\right)\Frf$, $\Hb[k]\triangleq [\hb_1[k],\ldots,\hb_{\Nu}[k]]$ and $\Db_k\triangleq\diag\left(t_1[k],\ldots, t_{\Nu}[k]\right)$ with $t_i[k]\triangleq q_i[k]\sqrt{w_i(r_i[k]+1)} $. The optimal Lagrangian multiplier $\eta_k$ can be easily found by the bisection search method based on the complementary slackness, i.e., $\eta_k\left( P_{\rm b}- \|\Frf \Fbb[k]\|^2\right)=0,\forall k$.

The overall proposed FP-based algorithm for solving problem \eqref{pb:relaxed HBF design in MU} is summarized in Algorithm \ref{alg:FP relaxed HBF design}, where the initial $\Frf$ can be randomly set to satisfy \eqref{eq:relaxed analog BF constraint MU}. Then, we can construct the effective baseband channel $\Bar{\Hb}[k]\triangleq \Frf^H \Hb[k]\Hb^H[k]\Frf$. By performing an SVD $\Bar{\Hb}[k]=\Bar{\Ub}[k]\Bar{\Sb}[k] \Bar{\Vb}[k]$ for each subcarrier, the digital precoder is initialized as $\Fbb[k]=\Bar{\Vb}[k]_{:,1:\Nu}$, whose columns are the $\Nu$ right singular vectors of $\Bar{\Hb}[k]$. Finally, $\Fbb[k]$ is normalized to ensure the transmit power constraint, i.e., $\Fbb[k]=\frac{\sqrt{P_{\rm b}}\Fbb[k]}{\|\Frf \Fbb[k] \|}$. We note that Algorithm \ref{alg:FP relaxed HBF design} converges to at least a stationary point since the objective function \eqref{eq:fq} is monotonically nondecreasing after each iteration \cite{shen2018fractional}. 
\begin{algorithm}[hbpt]
\small
% \setstretch{1}
\caption{ FP-based algorithm for problem \eqref{pb:relaxed HBF design in MU}}\label{alg:FP relaxed HBF design}
\LinesNumbered %要求显示行号
% \KwIn{$i=0,\Fb_{\rm RF}^{(i)}\in \Fc,c$}
\KwOut{$\Fb_{\rm RF}^{\star}\in \Fcl, \{ \Fbb^{\star}[k]\}_{k=1}^K$}
Initialize $\Fb_{\rm RF}\in \Fcl$ and $\{ \Fbb [k]\}_{k=1}^K$ satisfying \eqref{eq: power constraint multiuser cases}\;
\Repeat{{\rm stopping criteria is satisfied}}{
  Update $\{ \rb^{\star}[k]\}_{k=1}^K $ by \eqref{eq:optimal rmk}.

 Update $\{ \qb^{\star}[k]\}_{k=1}^K $ by \eqref{eq:optimal qmk}.

Update $\Fb_{\rm RF}^{\star}$ as the solution to problem \eqref{pb:analog BF design MU} by the PGA algorithm.
 
  Update $ \{ \Fbb^{\star}[k]\}_{k=1}^K$ by \eqref{eq:CF digital BF compact MU}.
 }
 \end{algorithm}

With $ \{ \Fbb^{\star}[k]\}_{k=1}^K$ and $\Fb_{\rm RF}^{\star}\in \Fcl$ obtained by Algorithm \ref{alg:FP relaxed HBF design}, we can perform the TS procedure to find an efficient solution to the following problem 
\begin{subequations}\label{pb:TS target MU}
  \begin{align}
    \underset{ \Frf}{\text{max}} \quad  & \frac{1}{K} \sum_{k=1}^{K} \sum_{m=1}^{\Nu} w_m R_m[k] \label{eq:TS target function}\\
    \text{s.t.} \quad  &  \eqref{eq: analog precoder constraint multiuser cases}. \nonumber
\end{align}
\end{subequations}
The above problem can be solved with Algorithm \ref{alg:APGA-TS algorithm} by replacing its target objective function and PGA-enabled initialization with \eqref{eq:TS target function} and $\Fb_{\rm RF}^{\star}\in \Fcl$, respectively. %It is worth noting that Algorithm \ref{alg:APGA-TS algorithm} applied to problem \eqref{pb:TS target MU} has a significantly lower complexity compared to problem \eqref{eq:analog precoder design} because the former has a much simpler form of objective function. 
We note that applying the PGA-TS algorithm here does not require performing the PGA initialization since $\Fb_{\rm RF}^{\star}\in \Fcl$ is already given. We omit the detailed subsequent steps to avoid redundancy. Instead, we summarize the two-step algorithm for solving problem \eqref{pb:hybrid precoding in Multiuser cases} in Algorithm \ref{alg:two-step MU algorithm} and present the complexity analysis next.
\begin{algorithm}[hbpt]
\small
% \setstretch{1}
\caption{ Two-step algorithm for problem \eqref{pb:hybrid precoding in Multiuser cases}}\label{alg:two-step MU algorithm}
\LinesNumbered %要求显示行号
% \KwIn{$i=0,\Fb_{\rm RF}^{(i)}\in \Fc,c$}
\KwOut{$\Fb_{\rm RF}, \{ \Fbb [k]\}_{k=1}^K$}
Initialize $\Fb_{\rm RF}\in \Fcl$ and $\{ \Fbb [k]\}_{k=1}^K$ satisfying \eqref{eq: power constraint multiuser cases}.

Obtain $ \{ \Fbb^{\star}[k]\}_{k=1}^K$ and $\Fb_{\rm RF}^{\star}\in \Fcl$ by Algorithm \ref{alg:FP relaxed HBF design}.

Obtain $\Fb_{\rm RF}$ satisfying \eqref{eq: analog precoder constraint multiuser cases} by the modified Algorithm \ref{alg:APGA-TS algorithm} based on $ \{ \Fbb^{\star}[k]\}_{k=1}^K$ and $\Fb_{\rm RF}^{\star}\in \Fcl$.
\end{algorithm}

\subsection{Complexity Analysis}\label{sec:complexity analysis for two-setp MU alg}
We first evaluate the complexity of Algorithm \ref{alg:FP relaxed HBF design} for the typical setting $\Nrf,\Nu \ll \Nt$ in downlink MU-MISO systems. The complexities of computing $\{ \rb^{\star}[k]\}_{k=1}^K $ and $\{ \qb^{\star}[k]\}_{k=1}^K $ are approximately the same, which is $\Ocl(K\Nu^2\Nt\Nrf)$. To obtain $\Frf^{\star}$ in step 5, the matrices $\Xib$, $\{ \Yb_k\}_{k=1}^K$, and $\{ \Zb_k\}_{k=1}^K$ are required but only computed once before the PGA procedure. The resultant complexity is approximate $\Ocl (K\Nu\Nt^2)$. On the other hand, the complexity of the PGA algorithm is approximately $\Ocl(K\Nrf\Nt^2)$, which is mainly due to computing the gradient $\nabla_{\Frf} g\left( \Frf\right)$. Therefore, the overall complexity of step 5 is $\Ocl\left(K(N_{\rm pga}\Nrf+\Nu)\Nt^2\right)$, where $N_{\rm pga}$ represents the number of iterations of the PGA algorithm. The complexity of step 6 is approximate $\Ocl\left( K\Nrf^3 +K\Nrf(\Nrf+\Nu)\Nt \right)$ due to matrix multiplications. Hence, the overall complexity of Algorithm \ref{alg:FP relaxed HBF design} can be approximated as $\Icl_{\rm FP}\Ocl\left( K(N_{\rm pga}\Nrf+2\Nu)\Nt^2 \right)$, where $\Icl_{\rm FP}$ denotes the number of iterations in Algorithm \ref{alg:FP relaxed HBF design}.  

The complexity of step 3 in Algorithm \ref{alg:two-step MU algorithm} is given as $\Icl_{\rm TS}\Ocl(\Icl_{\rm N}K\Nrf\Nt\Nu^2)$, which comes mainly from computing the objective values \eqref{eq:TS target function} during the neighbor search. Here, $\Icl_{\rm N}$ denotes the size of the neighbor set satisfying $ \Icl_{\rm N}\leq N_{\rm T}N_{\rm RF}$, and $\Icl_{\rm TS}$ is the number of iterations of the PGA-TS algorithm. Therefore, the overall complexity of Algorithm \ref{alg:two-step MU algorithm} is approximately $\Ocl\left(\Icl_{\rm FP}K(N_{\rm pga}\Nrf+2\Nu)\Nt^2+ \Icl_{\rm TS}\Icl_{\rm N}K\Nrf\Nt\Nu^2\right)$. Finally, we summarize the complexities of proposed algorithms in Table~\ref{tb:Complexity of proposed algorithms},, where we recall $c\triangleq 4N_{\rm pga}+2N_{\rm iter}^{\rm max}\Nnb$. 
\begin{table}[htbp]
\small
\renewcommand\arraystretch{1.5}
\centering
\vspace{-0.2cm}
\caption{Computational Complexities of the proposed algorithms. }
\vspace{-0.2cm}
\label{tb:Complexity of proposed algorithms}
% \begin{threeparttable}
    \begin{tabular}{|c|c|} %l(left)居左显示 r(right)居右显示 c居中显示
    \hline 
      Item & Computational complexity\\
    \hline  
    Algorithm 1 & $\Ocl\left(4N_{\rm pga}KN_{\rm RF}N_{\rm T}^2 \right) $    \\%[5ex]
    \hline 
    Algorithm 2 &  $\Ocl\left(c KN_{\rm RF}N_{\rm T}^2\right)$  \\%[5ex]
        \hline 
    Algorithm 3 &  \makecell[c]{$\Ocl\big((c\Nrf+2\Nt)K\Nr^2$ \\$ \qquad +  (c\Nrf+2\Nr)K\Nt^2\big)$ }   \\%[5ex]
    \hline 
     Algorithm 4  &  $\Ocl\left( \Icl_{\rm FP}K(N_{\rm pga}\Nrf+2\Nu)\Nt^2 \right)$\\
    \hline
     Algorithm 5  & \makecell[c]{$\Ocl\left(\Icl_{\rm FP}K(N_{\rm pga}\Nrf+2\Nu)\Nt^2 \right.$\\$\left.+\Icl_{\rm TS}\Icl_{\rm N}K\Nrf\Nt\Nu^2\right)$}\\

    \hline
  
    \end{tabular}   
      %    \begin{tablenotes}    
      %   \footnotesize               
      %   \item[*] $\Ndt$ and $\Ndr$ denote the number of TTDs at each RF chain in the Tx and Rx, respectively. 
      %   \item[**] $\Nfdt$ and $\Nfdr$ denote the number of FTTDs at each RF chain in the Tx and Rx, respectively.
      % \end{tablenotes}          
    % \end{threeparttable}
    % \vspace{-0.3cm}
\end{table}
\vspace{-5mm}
\section{Simulation Results}\label{sc:simulation results}
In this section, we provide numerical results to evaluate the performance of the proposed SW-HBF schemes. In the simulations, unless otherwise stated, we set $L_{\rm p}=4, \Delta=\frac{1}{2},f_{\rm c}=300~{\rm GHz},K=128$, and $N_{\rm S}=N_{\rm RF}=4$ for SU-MIMO ($\Nu=\Nrf=4$ for MU-MISO). The AoA/AoDs are uniformly distributed over $\left[-\frac{\pi}{2},\frac{\pi}{2}\right]$, and the pulse shaping filter is modeled by the raised cosine function\cite{alkhateeb2016frequency} with a roll-off factor $1$.
% \begin{equation}
%     p(t)=
%     \begin{cases}
%         \frac{\pi}{4} \operatorname{sinc}\left(\frac{1}{2 \beta}\right), &\text{if~} t=\pm \frac{T_{\mathrm{s}}}{2 \beta} \\
%         \operatorname{sinc}\left(\frac{t}{T_{\mathrm{s}}}\right) \frac{\cos \left(\frac{\pi \beta t}{T_{\mathrm{s}}}\right)}{1-\left(\frac{2 \beta t}{T_{\mathrm{s}}}\right)^{2}}, &\text {otherwise},
%     \end{cases}
%     \nthis \label{eq:pulse shaping}
% \end{equation}
%where $T_{\rm s}$ is the sampling period, and the roll-off factor $\beta$ is set to $\beta=1$. 
The path delay is uniformly distributed in $[0,(\varsigma -1)T_{\rm s}]$ where $\varsigma$ is the cyclic prefix length, given by $\varsigma=K/4$ according to 802.11ad specification \cite{alkhateeb2016frequency}. Other parameters are specified as given in each figure. The SNR is defined as SNR$~\triangleq \frac{P_{\rm b}}{\sigma_{\rm n}^2}$.  All reported results are averaged over $10^3$ channel realizations.

\subsection{Convergence of the Proposed Algorihthms}
% \vspace{-5mm}
% Figs.\ \ref{fig:Convergence of PGA}--\ref{fig:Convergence of APGA-TbRs}
  %  \begin{figure}[t]
  %  \vspace{-0.3cm}
  %       \centering	
  %        \includegraphics[width=0.4\textwidth]{convergence_all.eps}
  %        \captionsetup{font={small}}
  %        \vspace{-0.3cm}
  %       \caption{Convergence of the proposed algorithms with $N_{\rm T}=N_{\rm R}=64,N_{\rm S}=N_{\rm RF}=4$ for SU-MIMO and $N_{\rm T}=64,\Nu=4$ for MU-MISO}
  %       	\label{fig:Convergence property}
  %       	\vspace{-0.1cm}
  % \end{figure}
    \begin{figure}[t]
  % \vspace{-0.3cm}
        \centering
        \hspace{-2mm}
      \subfigure[]
         {\label{fig:convergence_all}\includegraphics[width=0.24\textwidth]{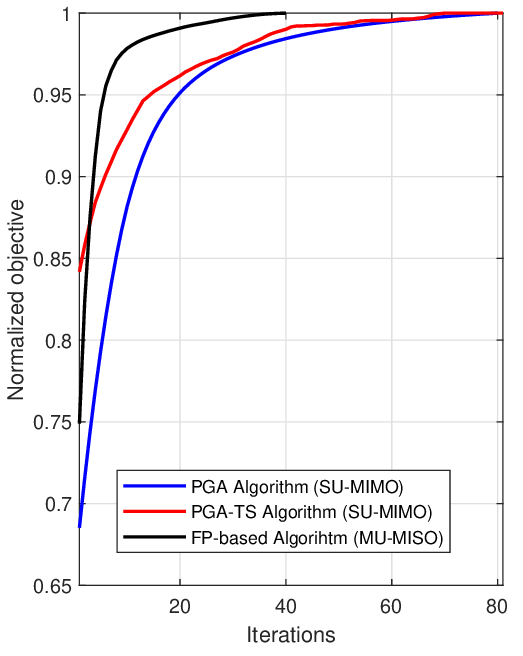}} \hspace{-2mm}
        \subfigure[]
        {\label{fig:convergence_PGATS_Nnb} \includegraphics[width=0.24\textwidth]{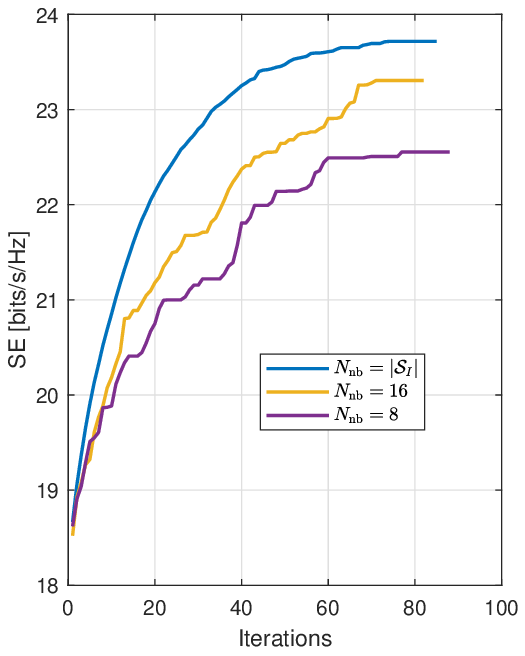}}
        \captionsetup{font={small}}
        \vspace{-2mm}
        \caption{ Convergence of the proposed algorithms with $N_{\rm T}=N_{\rm R}=64,N_{\rm S}=N_{\rm RF}=4$ for SU-MIMO and $N_{\rm T}=64,\Nu=4$ for MU-MISO. In Fig.~(a), the size of the neighbor set for the PGA-TS algorithm is set as $\Nnb=|\Scl_I|$. In Fig.~(b), we show its convergence with $\Nnb\in \{8,16,|\Scl_I|\}$.}
        \label{fig:Convergence property} 
        \vspace{-2mm}
  \end{figure}
Fig.\ \ref{fig:convergence_all} shows the convergence of Algorithms \ref{alg:PGA algorithm}, \ref{alg:APGA-TS algorithm}, and \ref{alg:FP relaxed HBF design} with the normalized objective values. In simulations, we set the termination conditions for Algorithm~\ref{alg:APGA-TS algorithm} with $\Nnb=|\Scl_I|$, $N_{\rm end}=10$ and $N_{\rm iter}^{\rm max}=200$. As a result, the PGA-TS algorithm is terminated if the objective value remains unchanged over $N_{\rm end}$ consecutive iterations, or if the number of iterations exceeds $N_{\rm iter}^{\rm max}$. It can be observed that the PGA algorithm, PGA-TS algorithm, and the FP-based algorithm can converge after around $40$, $70$, and $70$ iterations, respectively. In Fig.~\ref{fig:convergence_PGATS_Nnb}, we show the convergence of the PGA-TS algorithm with $\Nnb\in \{8,16,|\Scl_I|\}$. The smaller neighbor set is obtained by randomly selecting $\Nnb$ neighbor points from the full neighbor set. It can be observed that a smaller size of the neighbor set leads to a lower SE, though it reduces computational complexity. Additionally, examining fewer neighbor points lowers the likelihood of discovering a better solution, resulting in the observed staircase pattern in SE improvement over iterations. In subsequent figures, we set $\Nnb=|\Scl_I|$ unless specified otherwise.

\subsection{Performance of SW-HBF in SU-MIMO Systems}\label{sec:simulations for SU-MIMO}
In Fig.\ \ref{fig:SE v.s. Nt es}, we compare the SE attained by the proposed SW-HBF schemes with those achieved by the random and ES schemes with $N_{\rm S}=N_{\rm RF}=2$ and various number of antennas ($N_{\rm T}=N_{\rm R}$). Specifically, in the random SW-HBF scheme, the connections between the RF chains and antennas are randomly generated, i.e., the entries of the analog beamformers are randomly set as $0$ or $1$. In contrast, the ES scheme examines all possible combinations of those connections to search for the best one, which requires a long run time and extremely high complexity. Therefore, we consider a system of moderate size with up to 12 antennas. It can be seen from Fig.\ \ref{fig:SE v.s. Nt es} that the proposed PGA-TS algorithm performs close to the ES while requiring significantly lower complexity than the latter, verifying its efficiency. In the subsequent results, we will omit the ES scheme due to its extremely high complexity and focus on large-sized systems. 

  \begin{figure*}[htbp]
  % \vspace{-0.3cm}
        \centering
        \hspace{-5mm}
        \subfigure[]
        {\label{fig:SE v.s. Nt es} \includegraphics[width=0.33\textwidth]{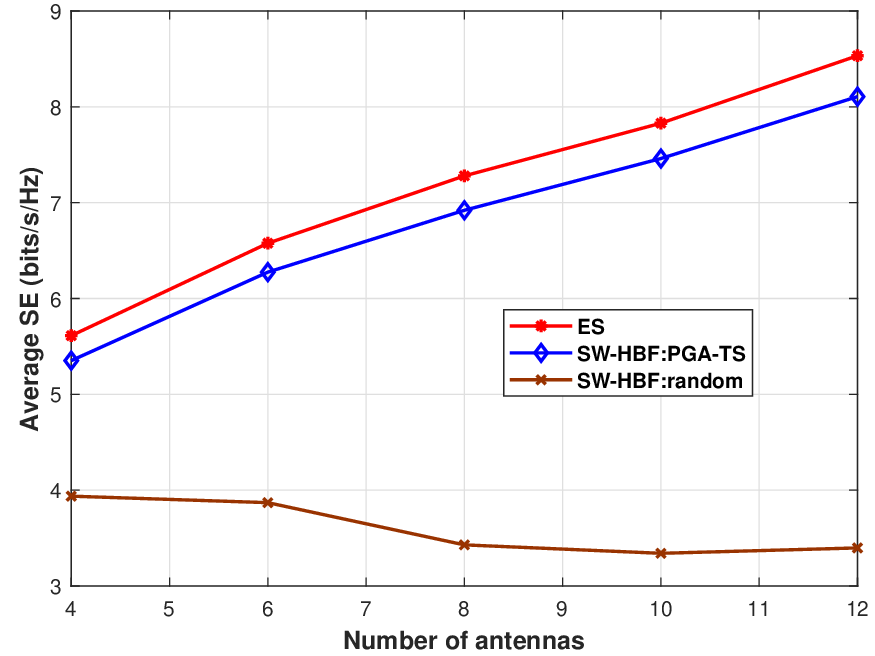}}
      \subfigure[]
         {\label{fig:SE vs SNR}\includegraphics[width=0.33\textwidth]{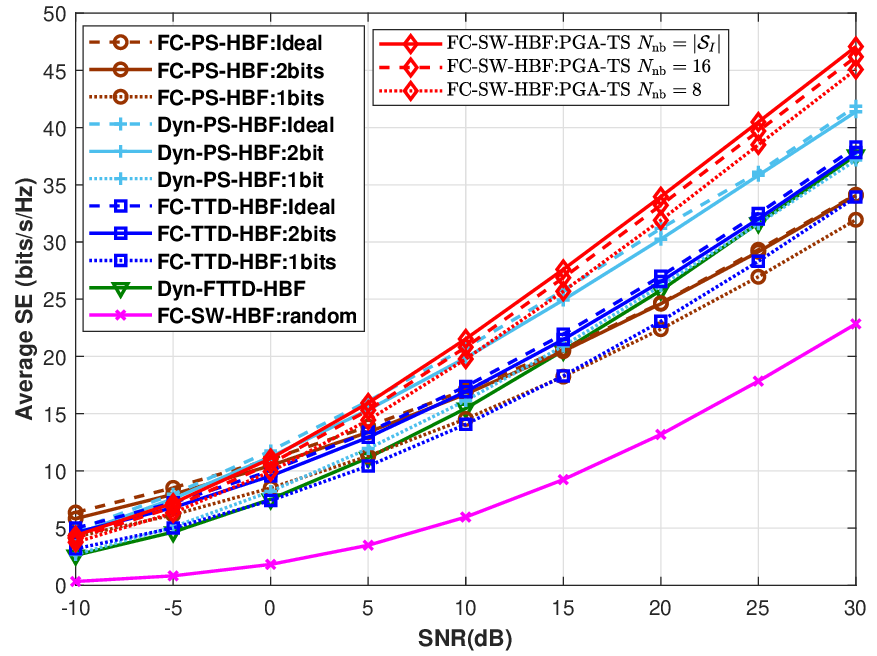}} 
        \subfigure[]
        {\label{fig:SE v.s. Ns} \includegraphics[width=0.33\textwidth]{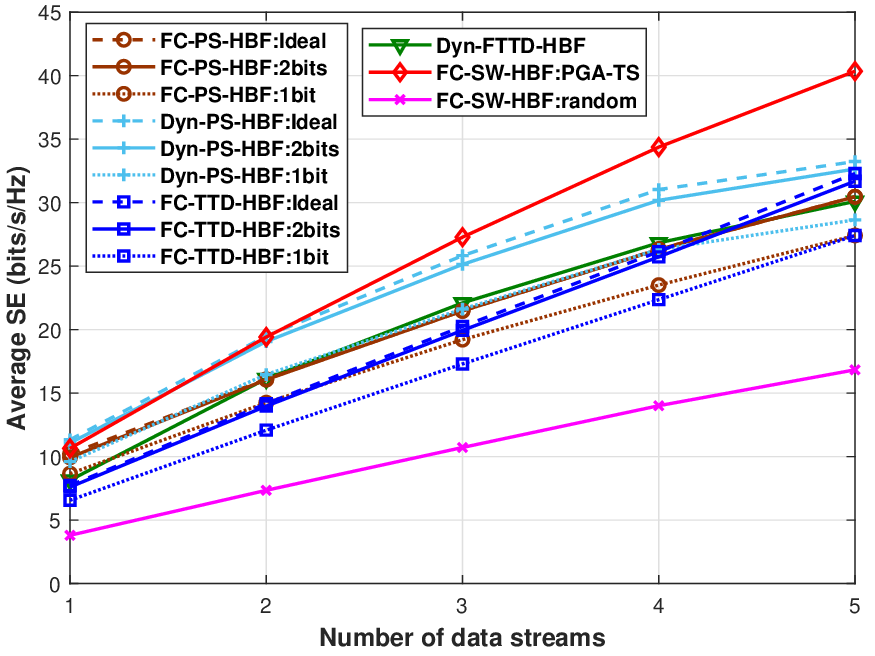}}
        \captionsetup{font={small}}
        \vspace{-0.3cm}
        \caption{ Achievable SE of SW-HBF schemes with $\fc=300~\ghz, B=30~{\rm GHz}$. Fig. (a) shows SE versus the number of antennas with $N_{\rm T}=N_{\rm R}$ and $N_{\rm S}=N_{\rm RF}=2$. Fig. (b) shows SE versus SNR with $N_{\rm S}=N_{\rm RF}=4$ and $N_{\rm T}=N_{\rm R}=256$. Fig. (c) shows SE versus the number of data streams with $N_{\rm RF}=N_{\rm S}$ and $\SNR=20~\dB,N_{\rm T}=N_{\rm R}=256$.}
        \label{fig:Achievable SE vs SNR and Ns} 
        	\vspace{-6mm}
  \end{figure*}  

% \subsection{Performance Comparison with Benchmark}

We next compare our proposed schemes to four state-of-the-art HBF designs, including the fully-connected PS-HBF (FC-PS-HBF) \cite{ma2021closed}, fully connected TTD-aided HBF (FC-TTD-HBF) \cite{dai2022delay}, dynamically connected PS-HBF (Dyn-PS-HBF) \cite{li2020dynamic}, and dynamically connected FTTD HBF (Dyn-FTTD-HBF) \cite{yan2022energy}. Note that the FC-PS-HBF employs only PSs, while the FC-TTD-HBF uses both PSs and TTDs. In contrast, Dyn-PS-HBF utilizes a combination of SWs and PSs, and Dyn-FTTD-HBF uses a combination of SWs and FTTDs. %Moreover, the former two schemes allow each RF chain to access all antennas while each RF chain only accesses a subset of all antennas for the latter two designs.

Because the works \cite{dai2022delay,yan2022energy} focus on only the precoder design, we modify the algorithms therein for combiner designs. Furthermore, since the FC-PS-HBF, Dyn-PS-HBF, and FC-TTD-HBF designs employ a large number of PSs, which can severely burden their power consumption, we adopt $\{1,2, \infty\}$-bit PSs for the three schemes in the following simulations for a fair comparison, where $\infty$-bit PSs indicates the ideal infinite-precision PSs. In subsequent figures, we designate the HBF schemes with the infinite-precision PSs as the "Ideal" cases. Fig. \ref{fig:SE vs SNR} shows the SE versus SNR with $N_{\rm T}=N_{\rm R}=256$ and $B=30~{\rm GHz}$, wherein FC-SW-HBF refers to the SW-HBF architecture studied in this work for clarity. The resultant BSR of the system according to Proposition \ref{remark:BSR equation} is $1.6$, indicating the severe beam squint effect. We set the number of TTDs and FTTDs as one and two for each RF chain here and in the subsequent results, except Fig.\ \ref{fig:SE v.s. EE} wherein the SE and EE of more TTDs and FTTDs will be demonstrated. We can observe from Fig.~\ref{fig:SE vs SNR} that PS-based HBF schemes with ideal PSs achieve similar SE compared to those with 2-bit PSs. This occurs because the phase noise becomes a minor issue when the beam squint effect is relatively severe. Furthermore, it is observed that the proposed SW-HBF schemes perform better than the benchmark schemes from medium to high SNRs. For example, at SNR $= 20$ dB, the PGA-TS SW-HBF scheme achieves SE improvements by $38\%$, $28\%$, and $12\%$ compared to the FC-PS-HBF, FC-TTD-HBF, and Dyn-PS-HBF schemes with $2$-bit PSs, respectively. Additionally, it achieves $30\%$ higher SE than that of the Dyn-FTTD-HBF design. Particularly, setting $\Nnb=\{8,16\}$ only result in minor SE loss for the proposed SW-HBF design compared to $\Nnb=|\Scl|$. The former achieves 94\% and 97\% the SE of the latter at $\SNR=20~\dB$ and still outperforms the benchmark HBF designs. The superiority of SW-HBF over other schemes can also be seen in Fig.\ \ref{fig:SE v.s. Ns}, which shows the SE versus the number of data streams with $\SNR=20~\dB$ and $N_{\rm RF}=N_{\rm S}$. %\blue{Furthermore, we can observe that HBF schemes with the ideal PSs performs similarly as with 2-bit ones. Due to significantly higher power consumption of the ideal PSs, HBF schemes with 2-bit PSs can result in higher EE. }. %\blue{We note that our simulations shows that the ideal phase-shifter HBF schemes achieve similar SE but signficantly lower EE compared to the 2-bit schemes. Considering the clarity of the figure, we do not present the performance of the ideal phase-shifter schemes} However, due to significantly higher power consumption of the ideal PSs, HBF schemes with 2-bit PSs can result in higher EE. Therefore, we focus on comparing the performance of the proposed SW-HBF design with that of HBF benchmarks with 2-bit PSs in the sequal.

We next compare both the SE and EE of the considered schemes. The EE is defined as the ratio between the SE and the total power consumption of the transceiver \cite{gao2016energy,li2020dynamic,dai2022delay}.  The power consumption of the considered schemes are presented in Table \ref{tb:power consumption of HBF architectures}, where DBF represents the conventional digital beamforming scheme. Note that $\Ndt$ ($\Nfdt$) and $\Ndr$ ($\Nfdr$) denote the number of TTDs (FTTDs) at each RF chain in the Tx and Rx, respectively. Furthermore, the power consumption of an RF chain is given as $P_{\rm RF}=P_{\rm M}+P_{\rm LO}+P_{\rm LPF}+P_{\rm BBamp}$ \cite{mendez2016hybrid}, where $P_{\rm M}$, $P_{\rm LO}$, $P_{\rm LPF}$, and $P_{\rm BBamp}$ are the power consumption of the mixer, the local oscillator, the low pass filter, and the baseband amplifier, respectively. The power consumptions of the low noise amplifier (LNA) at each receive antenna, power amplifier (PA) at each transmit antenna, and the analog-to-digital (digital-to-analog) converters at the Rx (Tx) are denoted as $P_{\rm LNA}$, $P_{\rm PA}$ and $P_{\rm ADC}$, respectively. Furthermore, we denote $P_{\rm common}\triangleq \Nt P_{\rm PA}+\Nr P_{\rm LNA} $ as the common power consumption for all architectures. Note that each RF chain requires two quantizers to quantize the in-phase and quadrature-phase signals separately \cite{ma2023analysis}. In addition, the power consumption of the splitter, combiner, TTD, FTTD, PS, and SW are represented by $P_{\rm SP}$, $P_{\rm C}$, $P_{\rm TTD}$, $P_{\rm FTTD}$, $P_{\rm PS}$ and $P_{\rm SW}$, respectively. Finally, the power consumptions of all components are given in Table \ref{Tb:power each device}. We note that the $P_{\rm PS} $ in Table \ref{tb:power consumption of HBF architectures} specifically refers to  the power consumption of 1-bit, 2-bit, or infinite-precision PSs, with values of $10$ mW, $20$ mW, and $40$ mW, respectively \cite{mendez2016hybrid,kim2014220,li2020dynamic}. 
 \begin{table*}
\small
\renewcommand\arraystretch{1.5}
\centering
\caption{Power consumption of FC-SW-HBF and benchmark transceiver architectures. }
\vspace{-0.2cm}
\label{tb:power consumption of HBF architectures}
% \begin{threeparttable}
    \begin{tabular}{c|c} %l(left)居左显示 r(right)居右显示 c居中显示
    \hline 
      Architecture & Power consumption\\
    \hline  
    DBF & $P_{\rm common}+ (N_{\rm T}+N_{\rm R})(P_{\rm RF}+2P_{\rm ADC})  $    \\%[5ex]
    \hline 
    FC-PS-HBF &  \makecell[c]{$P_{\rm common}+ (N_{\rm T}+N_{\rm R})N_{\rm RF}P_{\rm PS} +2N_{\rm RF}(P_{\rm RF}+2P_{\rm ADC})$ \\ $+ (\Nr+\Nrf)P_{\rm SP} +(\Nt + \Nrf)P_{\rm C}$}    \\%[5ex]
        \hline 
    Dyn-PS-HBF &  \makecell[c]{$P_{\rm common}+2N_{\rm RF}(P_{\rm RF}+2P_{\rm ADC})+ (N_{\rm T}+N_{\rm R})(P_{\rm SW}+P_{\rm PS}) $ \\ $+ \Nrf(P_{\rm SP}+P_{\rm C})$}    \\%[5ex]
    \hline 
     FC-TTD-HBF \tnote{*}  &  \makecell[c]{$P_{\rm common}+(N_{\rm T}+N_{\rm R})N_{\rm RF}P_{\rm PS} +2N_{\rm RF}(P_{\rm RF}+2P_{\rm ADC})+ \Nrf(\Ndt+\Ndr )P_{\rm TTD} $ \\ $ + (\Nr+\Nrf+ \Nrf \Ndt)P_{\rm SP} +(\Nt + \Nrf+ \Nrf \Ndr)P_{\rm C}$ }\\
    \hline
     Dyn-FTTD-HBF \tnote{**} & \makecell[c]{$P_{\rm common}+(N_{\rm T}+N_{\rm R})P_{\rm SW} +2N_{\rm RF}(P_{\rm RF}+2P_{\rm ADC}) $\\ $+ \Nrf (\Nfdt+\Nfdr )P_{\rm FTTD}+ \Nrf \Nfdt P_{\rm SP} +\Nrf \Nfdr P_{\rm C} $}\\
    \hline
    FC-SW-HBF & \makecell[c]{$P_{\rm common}+(N_{\rm T}+N_{\rm R})N_{\rm RF}P_{\rm SW} +2N_{\rm RF}(P_{\rm RF}+2P_{\rm ADC}) $ \\ $+ (\Nr+\Nrf)P_{\rm SP} +(\Nt + \Nrf)P_{\rm C} $} \\
    \hline
    
    \end{tabular}   
      %    \begin{tablenotes}    
      %   \footnotesize               
      %   \item[*] $\Ndt$ and $\Ndr$ denote the number of TTDs at each RF chain in the Tx and Rx, respectively. 
      %   \item[**] $\Nfdt$ and $\Nfdr$ denote the number of FTTDs at each RF chain in the Tx and Rx, respectively.
      % \end{tablenotes}          
    % \end{threeparttable}
    % \vspace{-0.3cm}
\end{table*}

\begin{table*}
% \vspace{-0.5cm}
\small
\renewcommand\arraystretch{1}
%\vspace{-0.3cm}
    \centering
    %\left
    \caption{Power consumption of component devices.}\label{Tb:power each device}
    \vspace{-0.2cm}
        \begin{tabular}{rlr|| rlr}
            \hline
            Device &  Notation   & Value & Device &  Notation   & Value \\
            \hline
            Low Noise Amplifier (LNA) \cite{singh2023design} & $P_{\rm LNA}$ &$136$ mW &  Splitter \cite{abbas2017millimeter}   & $P_{\rm SP}$ &$19.5$ mW\\           
            Power amplifier (PA) \cite{li2021250} & $P_{\rm PA}$ &$268$ mW & Switch \cite{mendez2016hybrid}          & $P_{\rm SW}$ & $5$ mW \\
            Combiner \cite{abbas2017millimeter}       & $P_{\rm C}$     &$19.5$ mW & Local oscillator \cite{mendez2016hybrid}  & $P_{\rm LO}$   & $5$ mW\\
            Mixer   \cite{mendez2016hybrid}         & $P_{\rm M}$     & $19$ mW & Base-band amplifier \cite{mendez2016hybrid} & $P_{\rm BBamp}$ & $5$ mW\\
            Low pass filter \cite{mendez2016hybrid}  & $P_{\rm LPF}$  & $14$ mW & True-time-delayer (TTD)\cite{cho2018true} & $P_{\rm TTD}$   & $285$ mW\\
            ADC/DAC \cite{greshishchev201960}            & $P_{\rm ADC}$   & $560$ mW & Fixed True-time-delayer (FTTD)\cite{yan2022energy} & $P_{\rm FTTD}$   & $63$ mW\\
            % 1-bit PS \cite{mendez2016hybrid,kim2014220,li2020dynamic} & $P_{\rm PS, 1bit} $ & 10 mW & 2-bit PS \cite{mendez2016hybrid,kim2014220,li2020dynamic}&$P_{\rm PS, 2bit} $ & 20 mW\\
            \hline
        \end{tabular}
        % \vspace{-0.3cm}
\end{table*}

  \begin{figure*}[htbp]
  \vspace{-1mm}
        % \centering
        \hspace{-5mm}
        \subfigure[Average SE versus BSR.]
        {\label{fig:SE v.s. BW} \includegraphics[width=0.33\textwidth]{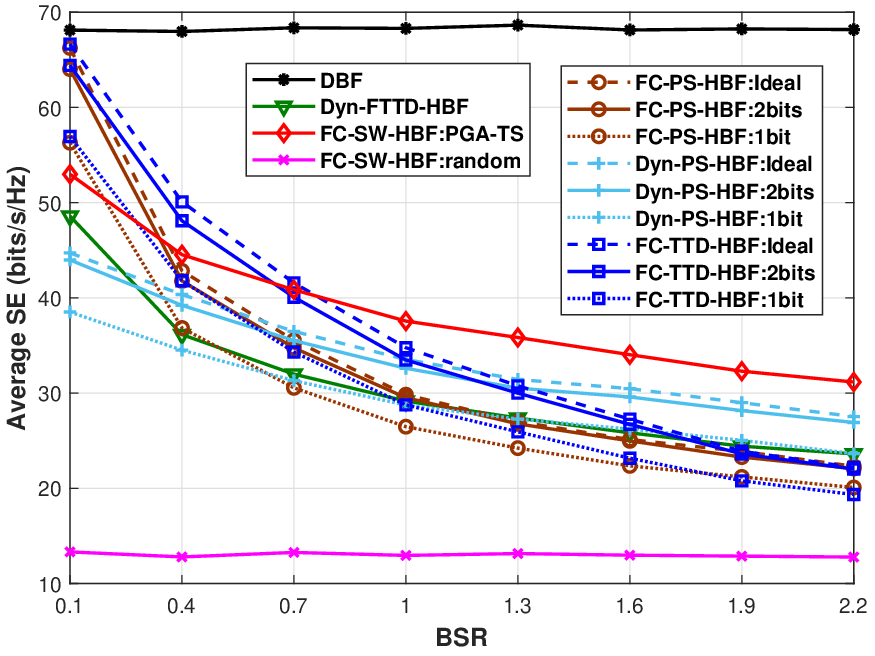}}
        \subfigure[Average EE versus BSR.]
        {\label{fig:EE v.s. BW}\includegraphics[width=0.33\textwidth]{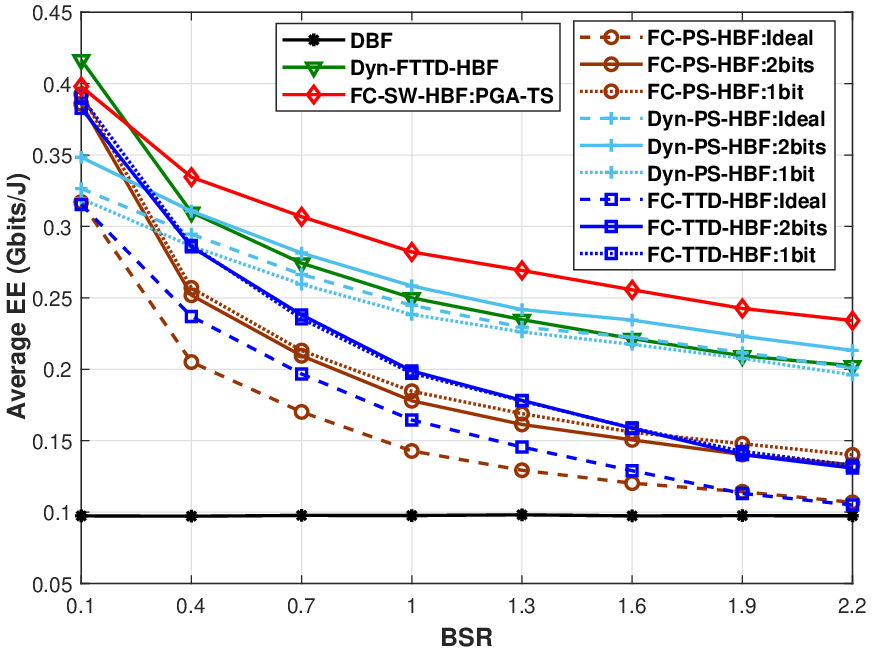}} 
        % \hspace{-10mm}
        \subfigure[Average SE versus average EE with $B=30~\ghz$,]
        {\label{fig:SE v.s. EE} \includegraphics[width=0.33\textwidth]{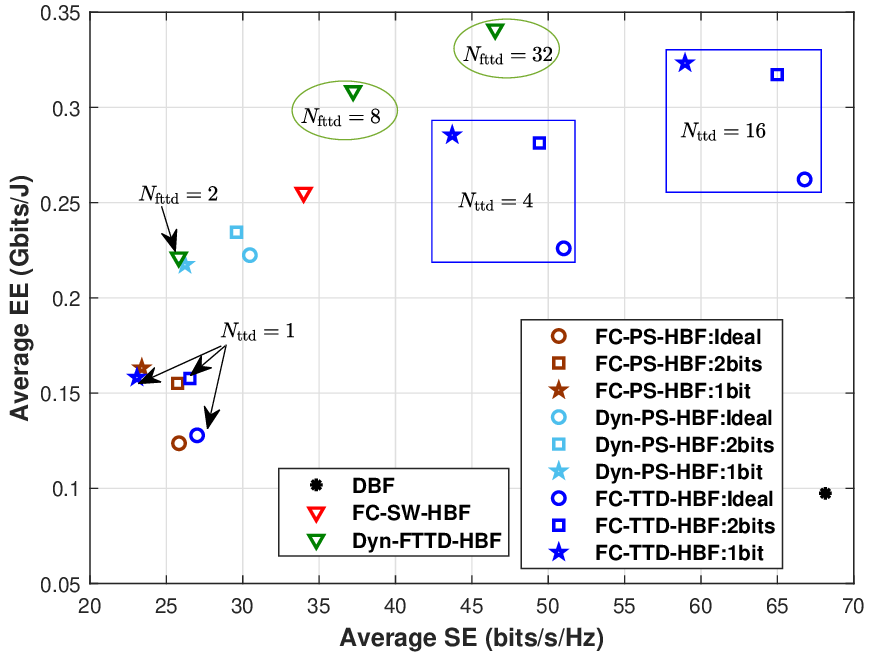}}
        \captionsetup{font={small}}
        \vspace{-3mm}
        \caption{ Average SE and average EE versus $\BSR=\{0.1,0.4,0.7,1,1.3,1.6,1.9,2.2\}$, which are obtained with bandwidth $B=\{1.875, 7.5, 13.125, 18.75, 24.375, 30, 35.625, 41.25\}~\ghz$ and $\fc=300~\ghz,N_{\rm T}=N_{\rm R}=256,N_{\rm S}=N_{\rm RF}=4,{\rm SNR}=20~{\rm dB}$.}
        \label{fig:SE an EE vs BSR} 
        	\vspace{-5mm}
  \end{figure*}  

  In Figs.\ \ref{fig:SE v.s. BW}--\ref{fig:EE v.s. BW}, we show the SE and EE of the considered schemes versus the BSR with $N_{\rm T}=N_{\rm R}=256, N_{\rm S}=N_{\rm RF}=4,{\rm SNR}=20~{\rm dB}$ and various system bandwidth. We note the following observations. Firstly, at $\text{BSR}=0.1$, the FC-PS-HBF scheme with $2$-bit resolution achieves up to $97\%$ the SE of the optimal DBF. Although the system has a large bandwidth $1.875~\ghz$ and a large number of antennas, the beam squint effect can be negligible, which confirms the findings in Section \ref{sc:beam squint effect analysis}. Secondly, it is seen for all the compared HBF schemes that the SE decreases as the BSR increases with the bandwidth. The FC-TTD-HBF designs with 2-bit and ideal PSs outperform other schemes when $\BSR \leq 0.7$. However, it performs worse than the proposed FC-SW-HBF design for increased bandwidth due to the more pronounced beam squint effect. Notably, the SE of the FC-PS-HBF schemes decreases more rapidly with increasing BSR compared to the FC-SW-HBF scheme. This observation indicates that the FC-SW-HBF design is more resilient to the beam squint effect than the FC-PS-HBF scheme. Consequently, the former exhibits significantly better performance than the latter. For instance, at $\BSR= 1.6$, the FC-SW-HBF design attains $38\%$ SE improvements compared to the $2$-bit FC-PS-HBF scheme. Furthermore, the higher power consumption of PSs makes the EE of the FC-PS-HBF scheme significantly lower than that achieved by the FC-SW-HBF scheme, as demonstrated in Fig.\ \ref{fig:EE v.s. BW}. Specifically, the FC-SW-HBF scheme attains $61\%$ EE improvements compared to the FC-PS-HBF scheme with $2$-bit PS at $\BSR= 1.6$. 
  
  The Dyn-PS-HBF design outperforms the FC-PS-HBF scheme but still performs worse than the proposed FC-SW-HBF design in terms of both SE and EE. This is because the Dyn-PS-HBF is more sensitive to the beam squint and has a lower array gain than the FC-SW-HBF scheme. Moreover, the FC-SW-HBF scheme also achieves higher SE and EE than the FC-TTD-HBF and Dyn-FTTD-HBF schemes for severe beam squint. Finally, we plot the SE--EE map of the considered schemes in Fig.\ \ref{fig:SE v.s. EE} wherein $N_{\rm td}$ ($N_{\rm ftd}$) represents the number of TTDs (FTTDs) in each RF chain (assuming $\Ndr=\Ndt=N_{\rm td}$ and $\Nfdr=\Nfdt=N_{\rm ftd}$). It can be observed that deploying more TTDs and FTTDs leads to a better SE--EE tradeoff, consistent with the findings in \cite{dai2022delay} and \cite{yan2022energy}. However, for limited TTDs and FTTDs, the proposed FC-SW-HBF design demonstrates its superiority over the benchmarks. %Due to the challenges of implementing TTDs discussed in Section \ref{sec:introduction prior works}, the FC-SW-HBF design thereby emerges as an attractive alternative for wideband large-scale MIMO systems.}

    \begin{figure}[htbp]
  % \vspace{-0.3cm}
        % \centering
        \hspace{-3mm}
      \subfigure[]
         {\label{fig:beam pattern vs ps}\includegraphics[width=0.24\textwidth]{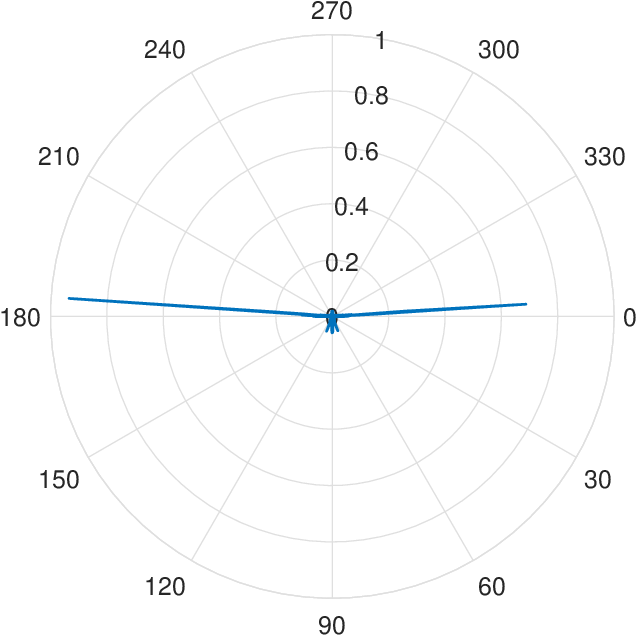}} \hspace{-1mm}
        \subfigure[ ]
        {\label{fig:beam pattern vs sw} \includegraphics[width=0.24\textwidth]{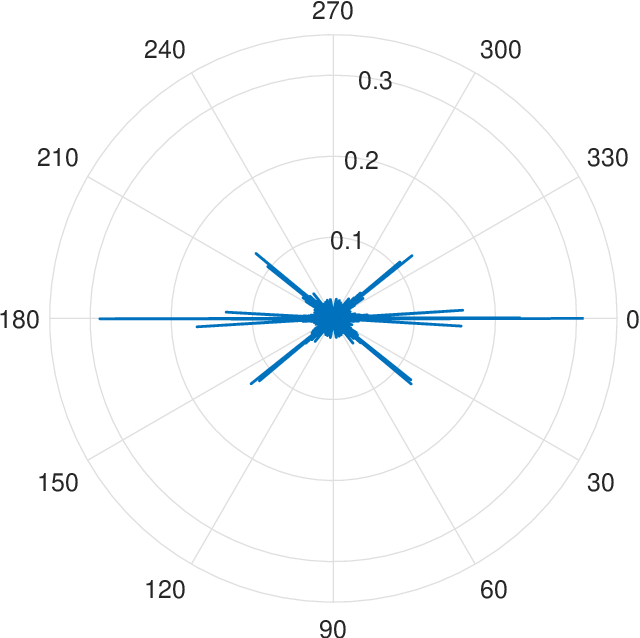}}
        \captionsetup{font={small}}
        \vspace{-0.3cm}
        \caption{  An example of beam pattern of analog beamformers with $\fc=300~\ghz,B=30~\ghz,N_{\rm T}=N_{\rm R}=256$, $\Ns=\Nrf=1$. Fig. (a) and Fig. (b) display the beam pattern of the PS-based and SW-based analog beamformer, respectively.}
        \label{fig:beam pattern} 
        	% \vspace{-0.3cm}
  \end{figure} 
  
Fig.~\ref{fig:beam pattern} shows an example of the beam pattern of PS-based and SW-based analog beamformers with $B=30~\ghz,N_{\rm T}=N_{\rm R}=256$, and $\Ns=\Nrf=1$. It can be observed that although having reduced beamforming gains, an SW-based analog beamforming architecture can generate more beams than the PS-based one. Under the severe beam squint effect, the former can cover the desired angles of incoming signals with a higher probability than the latter, yielding a higher EAG as analyzed in Section~\ref{sec:EAG of SW and PS array}. Therefore, the SW-HBF scheme is more resilient to the beam squint than the PS-HBF one. Consequently, the SW-HBF scheme can achieve a higher SE than the PS-HBF counterpart in the case of a severe beam squint and perform better at higher SNRs. Moreover, we observe that the complexity of the FC-PS-HBF design \cite{ma2021closed} is approximately $\Ocl(K\Nt^2\Nr)$, which is lower than that of the proposed FC-SW-HBF design. However, it is worth noting that the higher complexity of the proposed method, compared to the PS-based counterparts, does not necessarily imply inefficiency in the proposed algorithm. This is because the hardware constraints of the designs are different. Our simulation for a system with $\Nrf=\Ns=2$ and $\Nr=\Nt=12$ shows that the PGA-TS SW-HBF scheme achieves a runtime reduction of $ 2097100\%$ compared to the optimal exhaustive search. %Furthermore, when examining the runtimes of the considered benchmarks for $\Nrf=\Ns=2$ and $\Nr=\Nt=12$, it is shown that compared to the PGA-TS SW-HBF, the optimal ES-based SW-HBF requires $20971$ times more runtime, whereas the Dyn-PS-HBF, Dyn-TTD-HBF, FC-TTD-HBF, and FC-PS-HBF demand $17\%$, $11\%$, $1.7\%$, and $1.4\%$ of the runtime, respectively.}\blue{}

 \subsection{Performance of SW-HBF in MU-MISO Systems}\label{sec:MU simulations}
 \begin{figure*}[t]
  % \vspace{-0.5cm}
        \centering
        \hspace{-5mm}
        \subfigure[]
        {\label{fig:SE v.s. Ntes MU} \includegraphics[width=0.33\textwidth]{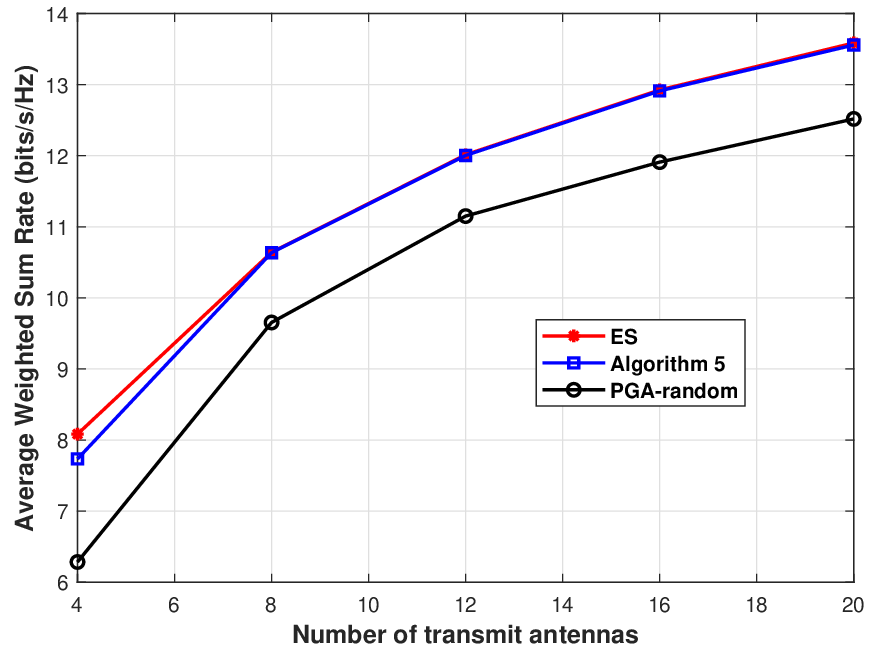}}
      \subfigure[]
         {\label{fig:SE vs SNR MU}\includegraphics[width=0.33\textwidth]{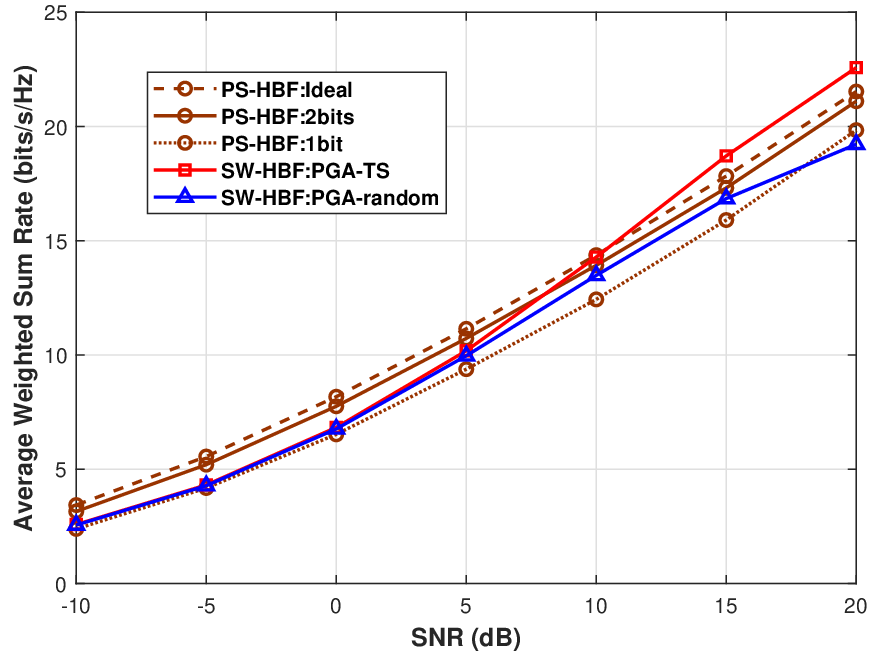}} 
        \subfigure[]
        {\label{fig:SE v.s. Nu MU} \includegraphics[width=0.33\textwidth]{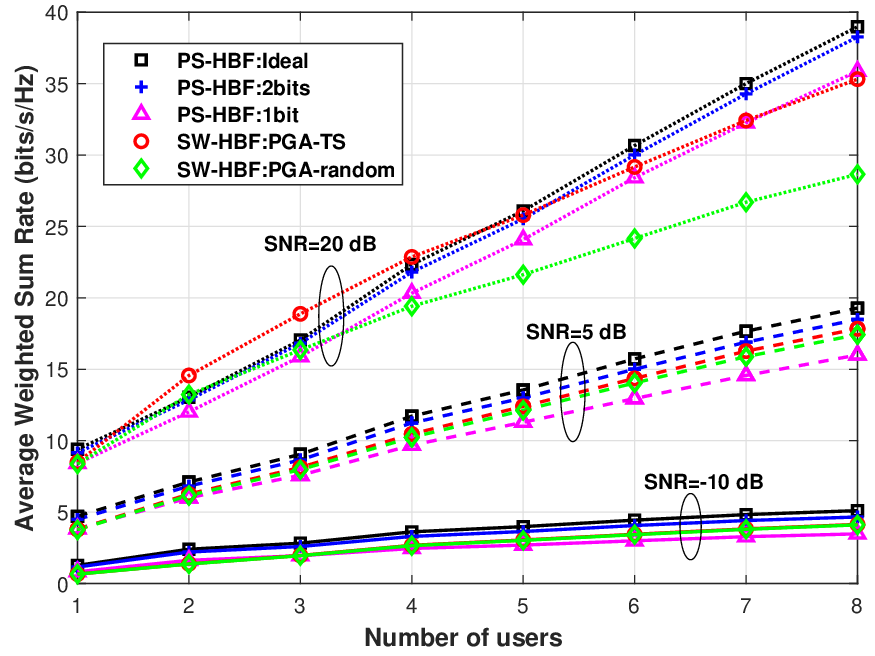}}
        \captionsetup{font={small}}
        \vspace{-0.3cm}
        \caption{ Average weighted sum rate of the SW-HBF scheme with $\fc=300~\ghz, B=30~{\rm GHz}$. Fig. (a) shows sum rate versus the number of transmit antennas ($\Nt$) with $\Nu=\Nrf=4$ and $\SNR=10~\dB$. Fig. (b) shows sum rate versus SNR with $N_{\rm T}=256$ and $\Nu=N_{\rm RF}=4$. Fig. (c) shows sum rate versus the number of users ($\Nu$) with $N_{\rm RF}=\Nu$ and $N_{\rm T}=256,\SNR=20~\dB$.}
        \label{fig:sum rate performance vs Nt SNR Nu} 
        	\vspace{-5mm}
  \end{figure*} 

    \begin{figure}[htbp]
   \vspace{-0.5cm}
        \centering	
         \includegraphics[width=0.35\textwidth]{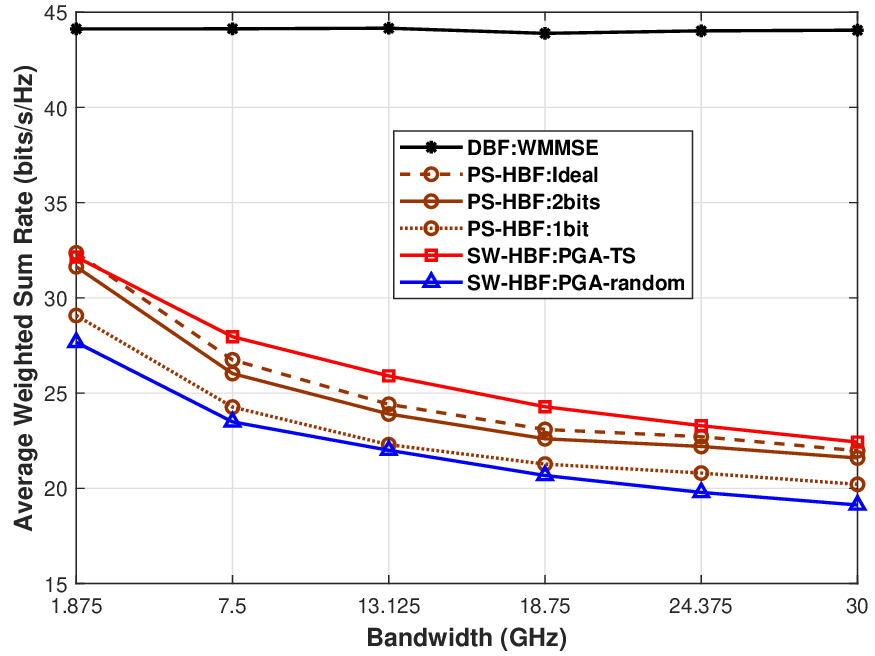}
         \captionsetup{font={small}}
         \vspace{-0.1cm}
        \caption{Average weighted sum rate versus bandwidth $B$ with $\fc=300~\ghz,N_{\rm T}=256,N_{\rm RF}=\Nu=4$ and ${\rm SNR}=20~{\rm dB} $.}
        	\label{fig:SE vs BSR MU}
        	\vspace{-0.1cm}
  \end{figure}
      \begin{figure}[htbp]
   % \vspace{-0.3cm}
        \centering	
         \includegraphics[width=0.35\textwidth]{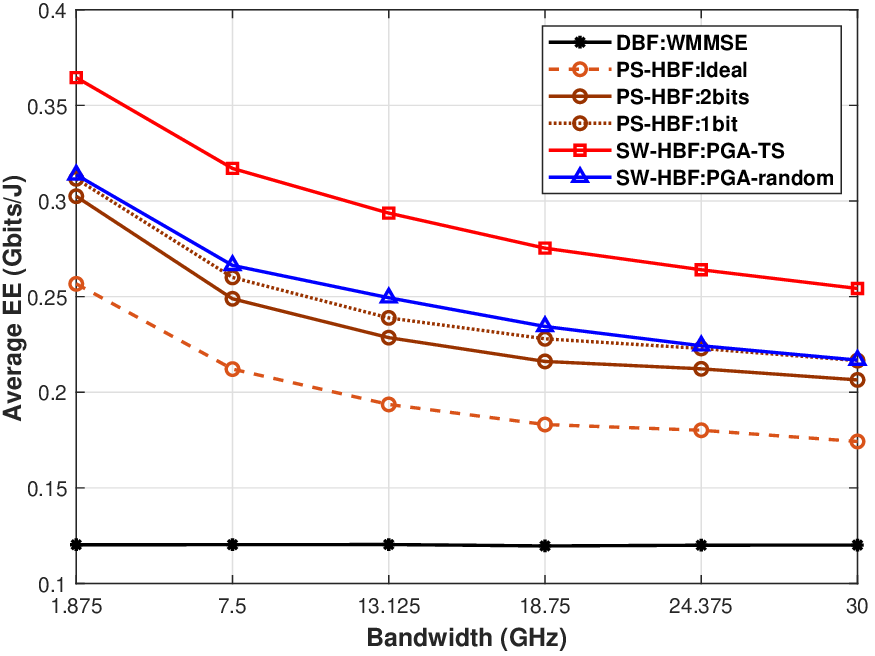}
         \captionsetup{font={small}}
         \vspace{-0.1cm}
        \caption{EE versus bandwidth $B$ with $\fc=300~\ghz,N_{\rm T}=256,N_{\rm RF}=\Nu=4$ and ${\rm SNR}=20~{\rm dB} $.}
        	\label{fig:EE vs BSR MU}
        	% \vspace{-0.3cm}
  \end{figure}
In this part, the channel from BS to each user at subcarrier $k$ is modeled as \cite{li2019hybrid}
    \begin{equation}\label{eq:d delay channel MU}
   \hb[k]=  \sqrt{\frac{N_{\rm T}}{L_{\rm p}}} \sum\limits_{l=1}^{L_{\rm p}} \alpha_l g_k(\tau_l)  \ab_{\rm t}\left( \theta_l^{\rm t}, f_k \right) , 
\end{equation}
where $g_k(\tau_l)\triangleq \sum\limits_{d=1}^{D-1}p\left(d T_{\rm s}-\tau_l\right)e^{-j\frac{2\pi k}{K}d} $. The channel parameters are set the same as those in Section \ref{sec:simulations for SU-MIMO} unless specified otherwise. Without loss of generality, we set the user weight $w_m=1, \forall m$ in subsequent simulations. 

We first show the efficiency of Algorithm \ref{alg:two-step MU algorithm} by comparing it with the ES method and PGA-random approach in Fig.\ \ref{fig:SE v.s. Ntes MU}. Specifically, the ES method examines all possible combinations of all RF chains and antennas based on the non-integer entries of $\Fb_{\rm RF}^{\star}\in \Fcl$ after running step 2 of Algorithm \ref{alg:two-step MU algorithm}. In contrast, the PGA-random method randomly sets the connection between all RF chains and antennas based on the non-integer entries of $\Fb_{\rm RF}^{\star}\in \Fcl$. It is seen in Fig.\ \ref{fig:SE v.s. Ntes MU} that Algorithm \ref{alg:two-step MU algorithm} performs close to the ES method and outperforms the PGA-random one. Moreover, the PGA-random approach performs relatively well, considering its low complexity. This is because the PGA solution $\Fb_{\rm RF}^{\star}\in \Fcl$ has only a few non-integer entries that need to be further optimized. Furthermore, the fewer non-integer entries lead to a small $\Icl_{\rm N}$, i.e., the size of the neighbor set in each TS iteration, which can significantly reduce the complexity and improve the convergence of PGA-TS procedure in step 2 of Algorithm \ref{alg:two-step MU algorithm}. We note that increasing the dimension of the analog precoder leads to a larger number of non-integer entries in the PGA solution $\Fb_{\rm RF}^{\star}\in \Fcl$. Consequently, the performance gap between Algorithm \ref{alg:two-step MU algorithm} and the PGA-random method becomes more pronounced as the number of antennas increases. 

In Figs.\ \ref{fig:SE vs SNR MU} and \ref{fig:SE v.s. Nu MU}, we show the average weighted sum rate versus SNR and the number of users $\Nu$, respectively. Here, Algorithm \ref{alg:two-step MU algorithm} is referred to as ``PGA-TS SW-HBF,'' while the PS-HBF is a benchmark developed in \cite{sohrabi2017hybrid}. It is observed that the proposed SW-HBF scheme can outperform the 2-bit PS-HBF from medium to high SNRs (i.e., $\SNR \geq 10 $ dB). For example, $7\%$ improvements of the average weighted sum rate are achieved by SW-HBF compared to the 2-bit PS-HBF at $\SNR=20$ dB. The superiority of the proposed SW-HBF design can also be justified in Fig.\ \ref{fig:SE v.s. Nu MU}, wherein the SW-HBF scheme achieves higher sum rate than the PS-HBF counterpart for $\Nu\in \{2,3,4\}$ at $\SNR=20$ dB. However, with more users or lower SNRs, the PS-HBF scheme marginally outperforms the proposed SW-HBF design. Nevertheless, the latter can achieve a higher EE than the former due to the substantially lower power consumption of switches compared to PSs. %However, we note that the SW-HBF design offers the advantage of a lower power consumption due to the utilization of switch networks. As a result, despite the marginally lower SE in certain scenarios, the SW-HBF approach can achieve a higher EE compared to the PS-HBF scheme.} 

Figs.\ \ref{fig:SE vs BSR MU} and \ref{fig:EE vs BSR MU} show the average weighted sum rate and EE versus system bandwidth, respectively, with $\Nt=256,\Nrf=\Nu=4$ and $\SNR=20~\dB$. In these simulations, we also present the performance of the DBF scheme that employs the WMMSE algorithm \cite{shi2011iteratively}.  %The EE is obtained by calculating the ratio between the weighted sum rate and the power consumption of BS hardware components presented in Section \ref{sec:simulations for SU-MIMO}. 
It is observed that even though the performance of both SW-HBF and PS-HBF schemes deteriorate with increasing bandwidth due to the beam squint effect, the former outperforms the latter in terms of both sum rate and EE. Moreover, although the DBF achieves the largest SE, it has the smallest EE due to the excessively high power consumption.

\vspace{-3mm}
\section{Conclusions}\label{sc:conclusions}
% \vspace{-10mm}
We have studied the performance of SW-HBF in wideband mmWave communications. We first present a closed-form expression of the BSR, which reveals that the beam squint effect linearly increases with the number of antennas, fractional bandwidth, and normalized antenna spacing. Moreover, we propose a transceiver design for SW-HBF in SU-MIMO OFDM systems and a hybrid precoding algorithm for MU-MISO OFDM systems. The numerical results verify that in both SU and MU scenarios, the proposed SW-HBF schemes exhibit higher resilience to beam squint. Thus, they outperform the compared HBF architectures in terms of SE and EE, especially when the beam squint effect becomes severe. These new findings make the SW-HBF a compelling solution in wideband systems. %Therefore, we conclude that the SW-HBF scheme can be an effective HBF strategy in THz communications when the beam squint effect is taken into account.} 
Future work may concern the imperfect channel information and hardware components, such as ADCs/DACs and PAs. Further investigation into the impact of spatial oversampling, including the effects of grating lobes, and the intelligent transmission surface-based transceiver are  important future research challenges. 
%show that have studied the beam squint effect and HBF design of the SW-HBF architecture in wideband large-scale MIMO systems. We obtained a closed-form BSR expression quantifying the beam squint effect, which reveals that the beam squint effect linearly increases with the number of antennas, fractional bandwidth, and normalized antenna spacing. Furthermore, through analyzing the EAG of the PS-based and SW-based arrays, we show that the SW-HBF architecture is more robust to beam squint than the PS-HBF counterpart. We have developed efficient SW-HBF algorithms to maximize the system SE. Specifically, we proposed the TS, PGA-TS, and PGA-TbRs algorithms for the analog beamforming design, while the digital ones are solved via closed-form solutions. The numerical results demonstrate that SW-HBF can significantly outperform PS-HBF in terms of both SE and EE when the beam squint effect is severe, as in wideband large-scale MIMO systems. Furthermore, we show that SW-HBF can achieve better SE or EE when the number of TTDs (FTTDs) is limited. For future works, we will study efficient channel estimation approaches and take low-resolution ADCs/DACs into consideration.
% \begin{appendices}
\vspace{-5mm}
 \appendices
 
\section{Proof of Proposition \ref{lemma:EAG of PS analog beamformer}}
\label{appd_arraygain_PS}
As it is difficult to obtain the integration in \eqref{eq:EAG for PS array}, we utilize the Newton polynomial to fit it by three points $(-1,g \left( -1,\xi_k\right) )$, $(0,g\left( 0,\xi_k\right))$, and $(1,g \left( 1,\xi_k\right))$, which leads to
\begin{align}\label{eq:Newton polynomial interpolation}%\hspace{-5mm}
      &\int^{1}_{-1}  g\left( \vartheta,\xi_k \right) {\rm d}\vartheta \nonumber \\
      &\int^{1}_{-1}  g\left( \vartheta,\xi_k \right) {\rm d}\approx \int^{1}_{-1} \left[\left(g \left( -1,\xi_k\right) -g\left( 0,\xi_k\right)\right)\vartheta^2+g\left( 0,\xi_k\right)\right] d\vartheta \nonumber \\
      &= \frac{4}{3}g\left( -1,\xi_k\right)  + \frac{2}{3}g\left( 0,\xi_k\right).
\end{align}
Since the normalized sinc function ${\rm sinc}(x)\triangleq \frac{\sin(\pi x)}{\pi x}$ achieves the maximum value 1 at $x=0$, i.e., $g\left( 0,\xi_k\right)=1$, the EAG of the PS-based array is expressed as $    \Es_{\rm ps}[g\left( \vartheta,\xi_k \right)]=\frac{2}{3K} \sum\limits_{k=1}^K g\left( -1,\xi_k\right) +\frac{1}{3}.$
% \begin{equation}\label{eq:EAQ of PS-based array in appendix}
%     \Es_{\rm ps}[g\left( \vartheta,\xi_k \right)]=\frac{2}{3K} \sum\limits_{k=1}^K g\left( -1,\xi_k\right) +\frac{1}{3}.
% \end{equation}
With \eqref{eq:xi bandwidth}, we obtain 
\begin{subequations}
 \begin{align}
    \Es_{\rm ps}[g\left( \vartheta,\xi_k \right)]&=\frac{2}{3K}\sum\limits_{k=1}^K\left|  \frac{\sin\left( N \Delta b \pi c_k   \right)}{N \sin\left(  \Delta b \pi c_k   \right)}\right|+ \frac{1}{3} \nonumber \\
    &\overset{(i)}{\approx} \frac{2}{3K}\sum\limits_{k=1}^K\left| {\rm sinc}\left(N \Delta b c_k \right)  \right|+ \frac{1}{3} \nonumber\\
    &\overset{(ii)}{\approx}\frac{2}{3} \int_0^1 \left|{\rm sinc}\left(N \Delta b \left(x-\frac{1}{2}\right)  \right) \right|dx+ \frac{1}{3} \nonumber\\
    &= \frac{4}{3 N\Delta b} \int_0^{\frac{N\Delta b}{2}} \left|{\rm sinc}( x)\right|dx + \frac{1}{3} \nonumber\\
  & \approx \frac{2}{ 3\text{~BSR}} \int_0^{\text{BSR}} \left|{\rm sinc}\left(4x\right)\right|dx + \frac{1}{3} \label{eq:BSR-EAG},
\end{align}
\end{subequations}
where $c_k\triangleq\frac{k}{K}-\frac{K+1}{2K}$; $(i)$ and $(ii)$ are due to $\left|\Delta b \pi c_k\right| \ll 1$ and $K \gg 1$, respectively. Equation \eqref{eq:BSR-EAG} is obtained with the $\text{BSR}$ in \eqref{eq:BSR closed-form}. Defining $f(z)\triangleq\frac{1}{z} \int_0^{z} \left|{\rm sinc}(4x)\right|dx$, we prove that $f(z)$ monotonically decreases with $z$ next. The first derivative of $f(z)$ is 
\begin{align}
\small
    f'(z)&=\frac{1}{z}\left( \left|{\rm sinc}\left(4z\right) \right|- \frac{1}{z} \int_0^{z} \left|{\rm sinc}\left( 4z\right) \right|dx\right) \nonumber\\
    &\overset{(i)}{=}\frac{1}{z}\left( \left|{\rm sinc}\left( 4z \right) \right|- \left|{\rm sinc}\left( 4\varepsilon\right) \right| \right),
\end{align}
where $(i)$ is due to the utilization of the {\it Mean Value Theorem} with $\varepsilon \in (0,z)$. 

 For $z \in (0,\frac{1}{4})$, the ${\rm sinc}\left( 4z\right)$ is non-negative and monotonically decreases, thus $f'(z)<0$.
 
 For $z \in [\frac{1}{4},\infty)$, because of $\left|{\rm sinc}(4z)\right|=\left|\frac{\sin(4\pi z)}{4\pi z}\right|\leq |\frac{1}{4\pi z}|=\frac{1}{4 \pi z}$ with $\left|\sin(4\pi z)\right|\leq 1$ and $z>0$, we obtain 
    \begin{align}
        f'(z) &\leq \frac{1}{z}\left( \frac{1}{4\pi z}- \frac{1}{z} \int_0^{z} \left|{\rm sinc}\left( 4x\right)\right|dx\right) \nonumber\\
        &=\frac{1}{z^2}\left( \frac{1}{4\pi}-\int_0^{\frac{1}{4}} {\rm sinc}\left( 4x\right)dx -\int_{\frac{1}{4}}^{z} \left|{\rm sinc}\left( 4x\right)\right|dx\right) \nonumber
    \end{align}
Since  $\frac{1}{4\pi}-\int_0^{\frac{1}{4}} {\rm sinc}\left( 4x\right)dx < 0$   and $\int_{\frac{1}{4}}^{z} \left|{\rm sinc}\left( 4x\right)\right|dx>0$, we obtain $ f'(z)<0$ for $z \in [\frac{1}{4},\infty)$.
    % \begin{equation}
    %     f'(z) \leq \frac{1}{z}\left( \frac{1}{\pi z}- \frac{1}{\pi \varepsilon  } \right) < 0
    % \end{equation}

In summary, the first derivative of $f(z)$ is negative in $(0,\infty)$, indicating that $f(z)$ monotonically decreaes with $z$ in $(0,\infty)$. Therefore, the EAG of the PS-based array monotonically decreases with BSR in \eqref{eq:BSR closed-form}. From \eqref{eq:BSR-EAG}, we can obtain $\frac{1}{3}\leq \Es_{\rm ps}[g\left( \vartheta,\xi_k \right)]\leq 1$. The lower bound is attained as $\text{BSR} \rightarrow \infty$, while the upper bound is achieved as $\text{BSR} \rightarrow 0$. The proof is completed.% \qedsymbol 
\vspace{-5mm}
\section{Proof of Proposition~\ref{prop:EAG of SW-based array}}\label{appdix:proof of SW EAG}
Define $\Omegab (\vartheta )$ as a $K\times N$ matrix with the $n$-th column given by $\frac{1}{K\sqrt{N }} \left[ e^{-j2\Delta (n-1)\pi \xi_1  \vartheta}, \ldots, e^{-j2\Delta (n-1)\pi \xi_K  \vartheta} \right]^T$. We have
\begin{align}\label{eq:EAG of switch bases array expansion}
    &\Es_{\rm sw}[g\left( \wb, \vartheta\right)] = \nonumber \\
    &\frac{1}{2K\sqrt{N \left\|\wb \right\|_1}}\int^{1}_{-1}  \sum_{k=1}^K   \left|\sum\limits_{n=1}^{N} w_n e^{-j2\pi(n-1)\Delta  \xi_k \vartheta }\right| d\vartheta \nonumber \\
    &=\frac{1}{2} \int^{1}_{-1} \tilde{g}(\wb,\vartheta)   d\vartheta,
\end{align}
where $\tilde{g}(\wb,\vartheta) \triangleq \frac{\left\|\Omegab (\vartheta )\wb \right\|_1}{\sqrt{\left\|\wb \right\|_1}}$. By using the Newton polynomial fitting the integration in \eqref{eq:EAG of switch bases array expansion} at three points $(-1,\tilde{g}(\wb,-1) )$, $(0,\tilde{g}(\wb,0)  )$ and $(1,\tilde{g}(\wb,1)  )$, we obtain
\begin{align}
    \tilde{g}(\wb,\vartheta) & \approx \left( \frac{\tilde{g}(\wb,1) +\tilde{g}(\wb,-1) }{2}-\tilde{g}(\wb,0)   \right)\vartheta^2  \nonumber \\
    &\phantom{\approx}\ +\frac{\tilde{g}(\wb,1)-\tilde{g}(\wb,-1)}{2}\vartheta +\tilde{g}(\wb,0) \nonumber \\
    &\overset{(i)}{=}\left( \tilde{g}(\wb,1) -\tilde{g}(\wb,0) \right)\vartheta^2+\tilde{g}(\wb,0),
\end{align}
where $(i)$ is due to $\tilde{g}(\wb,-\vartheta)= \tilde{g}(\wb,\vartheta)$. Therefore, we obtain
\begin{align}\label{eq:interpolated expected array gain sw}
&\Es_{\rm sw}[g\left( \wb, \vartheta \right)] \approx \medmath{\frac{1}{2}\int^{1}_{-1} \left(\left( \tilde{g}(\wb,1) -\tilde{g}(\wb,0) \right)\vartheta^2+\tilde{g}(\wb,0) \right) d\vartheta}  \nonumber\\  
&=\frac{1}{3}\tilde{g}(\wb,1) +\frac{2}{3}\tilde{g}(\wb,0) \overset{(ii)}{\approx} \frac{2}{3}\sqrt{\frac{\left\|\wb \right\|_1}{N}},
\end{align}
where approximation $(ii)$ follows 
\begin{align}\label{eq:switch EAG term1}
    &\tilde{g}(\wb,1)= \nonumber \\[-5pt]
    &
    \frac{1}{K\sqrt{N \left\|\wb \right\|_1}} \sum\limits_{k=1}^K \left|\sum\limits_{n=1}^{N} w_n e^{-j2\Delta\pi\left[1+\left(\frac{k}{K}-\frac{K+1}{2K} \right) b\right](n-1)}\right|
    \approx 0
\end{align}
by the {\it Law of Large Numbers} as $\left\|\wb \right\|_1 \gg 1$ in large-scale MIMO systems. Since 
$\left\|\wb \right\|_1 \leq N$, we have $\Es_{\rm sw}[g\left( \wb, \vartheta \right)] \leq \frac{2}{3}$, where the equality may hold when each RF chain is connected to all antennas. 

\vspace{-5mm}
\section{Proof of Proposition \ref{prop:FP equivalent expression}}\label{sec:appendix B}
By introducing $K$ auxiliary vectors $\rb[k]=[r_1[k],\ldots,r_{\Nu}[k]]^T,\forall k$, the objective value of problem \eqref{pb:hybrid precoding in Multiuser cases} can be achieved by the following problem
\begin{subequations}\label{pb:relaxed SINR sum-rate}
  \begin{align}
    \underset{\{ \rb[k]\}_{k=1}^K} {\text{max}} \; & \frac{1}{K} \sum_{(k,m)}  w_m \log_2(1+r_m[k])\\
    \text{s.t.} \;  &r_m[k] \leq\SINR_m[k] , \forall m,k, \label{eq:relaxed SINR}
\end{align}
\end{subequations}
where $\SINR_m[k] \triangleq \frac{\left|\hb_m^H[k]\Frf \fb_{{\rm BB}_m}[k] \right|^2}{\sum\limits_{i\neq m}^{\Nu}\left| \hb_m^H[k]\Frf \fb_{{\rm BB}_i}[k] \right|^2 + \sigma_{\rm n}^2}$.
This is because the equality holds for \eqref{eq:relaxed SINR} at the optima of problem \eqref{pb:relaxed SINR sum-rate}. By introducing $K$ multipliers $\lambdab[k]=[\lambda_1[k],\ldots, \lambda_{\Nu}[k]]^T,\forall k$, we can form a Lagrangian function as
\begin{align}\label{eq:Lagrangian function}
   & L(\{ \rb[k], \lambdab[k]\}_{k=1}^K)=\frac{1}{K} \sum_{(k,m)}  w_m \log_2(1+r_m[k]) \nonumber \\
    &\quad \quad -\sum_{(k,m)} \lambda_m[k]\left(r_m[k]-   \SINR_m[k]\right).
\end{align}
With $\frac{\partial L}{\partial \lambda_m[k]}=0$ and $\frac{\partial L}{\partial r_m[k]}=0$, we obtain the optimal $r_m[k]$ and $\lambda_m[k] $ as
\begin{align}
    &r_m^{\star}[k]=\SINR_m[k], \forall k,m, \label{eq:optimal r}\\
    &\lambda_m^{\star}[k]=\frac{w_m}{K \ln 2 \left(1+\SINR_m[k]\right)}, \forall k,m.  \label{eq:optimal lambda}
\end{align}
Inserting \eqref{eq:optimal lambda} back to \eqref{eq:Lagrangian function} and with some algebra, we obtain 
\begin{align}\label{eq:fr}
    &f_r(\Frf, \{ \Fbb [k], \rb[k]\}_{k=1}^K)= \nonumber \\
    &\frac{1}{K} \sum_{(k,m)}  w_m \log_2(1+r_m[k]) -\frac{1}{K\ln 2} \sum_{(k,m)}w_m r_m[k]\nonumber \\
    &  +\frac{1}{K\ln2} \sum_{(k,m)} \frac{w_m(r_m[k]+1)\left|\hb_m^H[k]\Frf \fb_{{\rm BB}_m}[k] \right|^2}{\sum\limits_{i} \left| \hb_m^H[k]\Frf \fb_{{\rm BB}_i}[k] \right|^2 + \sigma_{\rm n}^2}.
\end{align}
In this case, setting $\frac{\partial f_r}{\partial r_m[k]}=0$ yields \eqref{eq:optimal r}, with which $f_r(\Frf, \{ \Fbb [k], \rb^{\star}[k]\}_{k=1}^K)$ is same as the objective function of problem \eqref{pb:hybrid precoding in Multiuser cases}. It is observed that the last term of \eqref{eq:fr} has a multiple-ratio form, which is typical in FP problems. With the {\it quadratic transform} \cite{shen2018fractional}, we have
\begin{align}
   &\sum_{(k,m)} \frac{w_m(r_m[k]+1)\left|\hb_m^H[k]\Frf \fb_{{\rm BB}_m}[k] \right|^2}{\sum\limits_{i} \left| \hb_m^H[k]\Frf \fb_{{\rm BB}_i}[k] \right|^2 + \sigma_{\rm n}^2}=\nonumber\\
   &\max_{\{  \qb[k]\}_{k=1}^K} g(\qb[k]\}_{k=1}^K)
\end{align}
with 
\begin{align}\label{eq:quadratic transform}
   &g(\qb[k]\}_{k=1}^K)\triangleq  \nonumber \\[-3pt]
   &\sum_{(k,m)}2\sqrt{w_m(r_m[k]+1)}\Re\{q_m[k]^* \hb_m^H[k]\Frf \fb_{{\rm BB}_m}[k]\} \nonumber \\[-8pt]
   & \qquad -\sum_{(k,m)} |q_m[k]|^2 \left(\sum\limits_{i} \left| \hb_m^H[k]\Frf \fb_{{\rm BB}_i}[k] \right|^2 + \sigma_{\rm n}^2 \right),
\end{align}
where $\qb[k]=[q_1[k],\ldots,q_{\Nu}[k]]^T, \forall k$ are introduced $K$ auxiliary vectors. Setting $\frac{\partial g}{\partial q_m[k]}=0$ yields 
\begin{equation}\label{eq:optimal q}
    q_m^{\star}[k]=\frac{\sqrt{w_m(r_m[k]+1)}\hb_m^H[k]\Frf \fb_{{\rm BB}_m}[k]}{\sum\limits_{i} \left| \hb_m^H[k]\Frf \fb_{{\rm BB}_i}[k] \right|^2 + \sigma_{\rm n}^2}.
\end{equation}
Therefore, with the optimal $r_m[k]$ and $q_m[k]$ given as \eqref{eq:optimal r} and \eqref{eq:optimal q}, respectively, optimizing the objective function \eqref{eq:fr} is equivalent to maximizing the objective function \eqref{eq:fq} in the sense that they achieve the same optimal value and solutions. The proof is completed.% \qedsymbol
% \raggedright
\bibliographystyle{IEEEtran}
\bibliography{conf_short,jour_short,SW_HBF}

\end{document}